\documentclass[conference]{IEEEtran}
\IEEEoverridecommandlockouts
\usepackage{cite}
\usepackage{amsmath,amssymb,amsfonts}
\usepackage{amsthm}
\usepackage{textcomp}
\usepackage{xcolor}
\usepackage[misc]{ifsym}
\usepackage{graphics}
\usepackage{graphicx}
\usepackage{enumitem}
\usepackage{xspace}
\usepackage{caption}
\usepackage{subcaption}
\usepackage{bm}
\usepackage{mathrsfs}
\usepackage{array}
\usepackage{textcomp}
\usepackage{weiwAlgorithm}
\usepackage{url}
\usepackage{colortbl}
\usepackage{balance}
\usepackage{xparse}
\usepackage{multirow}
\usepackage{booktabs}
\usepackage{balance}
\usepackage{lipsum}

\newcommand{\revise}[1]{{#1}}

\newtheorem{example}{\textbf{Example}}
\newtheorem{theorem}{\textbf{Theorem}}
\newtheorem{corollary}{Corollary}
\newtheorem{definition}{\textbf{Definition}}
\newtheorem{lemma}{\textbf{Lemma}}

\makeatletter

\newcommand{\my@arrow}[1]{\ooalign{$#1-\mkern-5mu-$\cr\hidewidth$#1>$}}
\newcommand{\myarrow}{\mathrel{\mathpalette\my@arrow\relax}}

\newcommand\blfootnote[1]{%
  \begingroup
  \renewcommand\thefootnote{}\footnote{#1}%
  \addtocounter{footnote}{-1}%
  \endgroup
}

\begin{document}

\title{Minimizing the Influence of Misinformation \\via Vertex Blocking}

\author{
{
Jiadong Xie$^{\S\dagger}$, Fan Zhang$^{\star}$, Kai Wang$^{\sharp}$, Xuemin Lin$^{\sharp}$, Wenjie Zhang$^{\S}$}
\vspace{3.2mm}
\\
\fontsize{10}{10}
\selectfont\itshape
$^\S$University of New South Wales~
$^\dagger$Zhejiang Lab~
$^\star$Guangzhou University\\
$^\sharp$Antai College of Economics and Management, Shanghai Jiao Tong University
\fontsize{10}{10}
\selectfont\itshape
\\
\fontsize{9}{9} \selectfont\ttfamily\upshape
xiejiadong0623@gmail.com, zhangf@gzhu.edu.cn\\
\{w.kai,~xuemin.lin\}@sjtu.edu.cn, zhangw@cse.unsw.edu.au
}

\maketitle

\begin{abstract}

Information cascade in online social networks can be rather negative, e.g., the spread of rumors may trigger panic. To limit the influence of misinformation in an effective and efficient manner, the influence minimization (IMIN) problem is studied in the literature: given a graph $G$ and a seed set $S$, blocking at most $b$ vertices such that the influence spread of the seed set is minimized.
In this paper, we are the first to prove the IMIN problem is NP-hard and hard to approximate.
Due to the hardness of the problem, existing works resort to greedy solutions and use Monte-Carlo Simulations to solve the problem.
However, they are cost-prohibitive on large graphs since they have to enumerate all the candidate blockers and compute the decrease of expected spread when blocking each of them.
To improve the efficiency, we propose the AdvancedGreedy algorithm (AG) based on a new graph sampling technique that applies the dominator tree structure, which can compute the decrease of the expected spread of all candidate blockers at once.
Besides, we further propose the GreedyReplace algorithm (GR) by considering the relationships among candidate blockers.
Extensive experiments on 8 real-life graphs demonstrate that our AG and GR algorithms are significantly faster than the state-of-the-art by up to 6 orders of magnitude, and GR can achieve better effectiveness with time cost close to AG.
\end{abstract}





\blfootnote{$^*$ Fan Zhang is the corresponding author.}


\section{Introduction}
\label{sec:intros}

With the prevalence of social network platforms such as Facebook and Twitter, a large portion of people is accustomed to expressing their ideas or communicating with each other online. 
Users in online social networks receive not only positive information (e.g., new ideas and innovations) \cite{positive-inf}, but also negative messages (e.g., rumors and fake science) \cite{min-greedy2}. 
In fact, misinformation like rumors spread fast in social networks \cite{rumour-quick}, and can form more clusters compared with positive information \cite{anti-vacc}, which should be limited to avoid `bad' consequences. 
For example, the opposition to vaccination against SARS-CoV-2 (causal agent of COVID-19) can amplify the outbreaks \cite{covid19-2}. The rumor of White House explosions that injured President Obama caused a \$136.5 billion loss in the stock market~\cite{fakenews}. 
Thus, it is critical to efficiently minimize the influence spread of misinformation. 

We can model the social networks as graphs, where vertices represent users and edges represent their social connections. The influence spread of misinformation can be modeled as the expected spread under diffusion models, e.g., the independent cascade (IC) model~\cite{first-max}.
The strategies in existing works on spread control of misinformation can be divided into two categories: (i) blocking vertices \cite{min-greedy1,min-greedy-tree, FanLWTMB13, Nature-error, Viruses}, which usually removes some critical users in the networks such that the influence of the misinformation can be limited; or blocking edges \cite{min-edge1,min-edge2,WangDLYJY20,abs-1901-02156}, which removes a set of edges to stop the influence spread of misinformation; (ii) spreading positive information \cite{min-greedy2,LeeSMH19,seed-positive,HeSCJ12}, which considers amplifying the spread of positive information to fight against the influence of misinformation.

In this paper, we consider blocking key vertices in the graph to control the spread of misinformation.
Suppose a set of users are already affected by misinformation and they may start the propagation, we have a budget for blocking cost, i.e., the maximum number of users that can be blocked.
Then, we study the influence minimization problem \cite{min-greedy1,min-greedy-tree}: given a graph $G$, a seed set $S$ and a budget $b$, find a blocker set $B^*$ with at most $b$ vertices such that the influence (i.e., expected spread) from $S$ is minimized. Note that blocking vertices is the most common strategy for hindering influence propagation. For example, in social networks, disabling user accounts or preventing the sharing of misinformation is easy to implement. According to the statistics, Twitter has deleted 125,000 accounts linked to terrorism~\cite{twitter-delete}.
Obviously, we cannot block too many accounts, it will lead to negative effect on user experience. In such cases, it is critical to identify a user set with the given size whose blocking effectively hinders the influence propagation.



\vspace{1mm}
\noindent \textbf{Challenges and Existing Solutions.}
The influence minimization problem is NP-hard and hard to approximate, and we are the first to prove them (Theorems~\ref{theo::NP-IM} and~\ref{theo::APX}).
Due to the hardness of the problem, the state-of-the-art solutions use a greedy framework to select the blockers \cite{min-greedy1,min-greedy2}, which outperforms other existing heuristics \cite{min-greedy-tree,Nature-error,Viruses}.
However, different to the influence maximization problem, the spread function of our problem is not supermodular (Theorem~\ref{theo::expectedspread}), which implies that an approximation guarantee may not exist for existing greedy solutions.
\revise{Moreover, as the computation of influence spread under the IC model is \#P-hard~\cite{maximization1}, the state-of-the-art solutions use Monte-Carlo Simulations to compute the influence spread.
However, such methods are cost-prohibitive on large graphs since there are excessive candidate blockers and they have to compute the decrease of expected spread for every candidate blocker (detailed in Section~\ref{sec:ec-exist}).}


\vspace{1mm}
\noindent \textbf{Our Solutions.}
\revise{
Different to the state-of-the-art solutions (the greedy algorithms with Monte-Carlo Simulations), we propose a novel algorithm (GreedyReplace) based on sampled graphs and their dominator trees.
Inspired by reverse influence sampling \cite{Borgs-max}, the main idea of the algorithm is to simultaneously compute the decrease of expected spread of every candidate blocker, which uses almost a linear scan of each sampled graph.
We prove that the decrease of the expected spread from a blocked vertex is decided by the subtrees rooted at it in the dominator trees that generated from the sampled graphs (Theorem~\ref{theo:dominator-subtree}).
Thus, instead of using Monte-Carlo Simulations, we can efficiently compute the expected spread decrease through sampled graphs and their dominator trees.} 
We also prove the estimation ratio is theoretically guaranteed given a certain number of samples (Theorem~\ref{theorem:approx}). 
Equipped with above techniques, we first propose the AdvancedGreedy algorithm, which has a much higher efficiency than the state-of-the-art greedy method without sacrificing its effectiveness.

Furthermore, for the vertex blocking strategy, we observe that all out-neighbors of the seeds will be blocked if the budget is unlimited, while the greedy algorithm may choose the vertices that are not the out-neighbors as the blockers and miss some important candidates. {We then propose a new heuristic, named the GreedyReplace algorithm, focusing on the relationships among candidate blockers: we first consider blocking vertices by limiting the candidate blockers in the out-neighbors, and then try to greedily replace them with other vertices if the expected spread becomes smaller.}

\vspace{1mm}
\noindent \textbf{Contributions.}
Our principal contributions are as follows.

\begin{itemize}

\item 
     We are the first to prove the Influence Minimization problem is NP-hard and APX-hard unless P=NP. 
     
\item 

\revise{
We propose the first method to estimate the influence spread decreased by every candidate blocker under IC model, which only needs a simple scan on the dominator tree of each sampled graph.  
We prove an estimation ratio is guaranteed given a certain number of sampled graphs.
To the best of our knowledge, we are the first to study the dominator tree in influence related problems.
}

    
\item  \revise{Equipped with the above estimation technique, our AdvancedGreedy algorithm significantly outperforms the state-of-the-art greedy algorithms in efficiency without sacrificing effectiveness. We also propose a superior heuristic, the GreedyReplace algorithm, to further refine the effectiveness. }


\item Comprehensive experiments on $8$ real-life datasets validate that our AdvancedGreedy algorithm is faster than the state-of-the-art (the greedy algorithm with Monte-Carlo Simulations) by more than $3$ orders of magnitude, and our GreedyReplace algorithm can achieve better result quality (i.e., the smaller influence spreads) and close efficiency compared with our AdvancedGreedy algorithm. 
\end{itemize}
\section{Related Work}
\label{sec:relate}

\noindent \textbf{Influence Maximization.}
The studies of influence maximization are surveyed in~\cite{DBLP:journals/computing/AghaeeGBBF21,DBLP:journals/kais/BanerjeeJP20}.
Domingos et al. first study the influence between individuals for marketing in social networks \cite{value-customer}. Kempe et al. first model this problem as a discrete optimization problem \cite{first-max}, named Influence Maximization (IMAX) Problem. They introduce the independent cascade (IC) and linear threshold (LT) diffusion models, and propose a greedy algorithm with $(1-1/e)$-approximation ratio since the function is submodular under the above models. Borgs et al. propose a different method based on reverse reachable set for influence maximization under the IC model \cite{Borgs-max}. Tang et al. propose an algorithm based on martingales for IMAX problem, with a near-linear time cost \cite{max-sota}. 

\vspace{1mm}
\noindent \textbf{Influence Minimization.}
Compared with IMAX problem, there are fewer studies on controlling the spread of misinformation, as surveyed in~\cite{DBLP:journals/jnca/ZareieS21}.
Most works consider proactive measures (e.g., blocking nodes or links) to minimize the influence spread, motivated by the feasibility on structure change for influence study \cite{struct-sn,experience-min1,user-attribute2}. 
In real networks, we may use a degree based method to find the key vertices \cite{Nature-error,Viruses}.
Yao et al. propose a heuristic based on betweenness and out-degree to find approximate solutions \cite{YaoSZWG15}.
Wang et al. propose a greedy algorithm to block a vertex set for influence minimization (IMIN) problem under IC model \cite{min-greedy1}.
Yan et al. also propose a greedy algorithm to solve the IMIN problem under different diffusion models, especially for IC model \cite{min-greedy-tree}. 
They also introduce a dynamic programming algorithm to compute the optimal solution on tree networks. 
The above studies on the IMIN problem validate that the greedy heuristic is more effective than other methods, e.g., degree based heuristics \cite{min-greedy1,min-greedy-tree}.

Kimura et al. propose to minimize the dissemination of negative information by blocking links (i.e., finding $k$ edges to remove) \cite{min-edge1}. They propose an approximate solution for rumor blocking based on the greedy heuristic.
Other than IC model, the vertex and edge interdiction problems were studied under other diffusion models: \cite{min-edge2,NguyenCVD20,KhalilDS14} consider the LT (Linear Threshold) model, \cite{TongPEFF12} considers the SIR (Susceptible-Infected-Recovery) model and \cite{MedyaSS22,abs-1901-02156} considers CD (Credit Distribution) Model.

In addition, there are some other strategies to limit the influence spread. Budak et al. study the simultaneous spread of two competing campaigns (rumor and truth) in a network \cite{min-greedy2}. They prove this problem is NP-hard and provide a greedy algorithm which can be applied to the IMIN problem.
Manouchehri et al. then propose a solution with theoretically guaranteed for this problem~\cite{ManouchehriHD21}.
Moreover, Chen et al. propose the profit minimization of misinformation problem, which not only considers the number of users but also focus on interaction effects between users. As interaction effects are different between different users and the related profit obtained from interaction activities may also be different \cite{ChenLFGD19}.
Lee et al. also consider that both positive and negative opinions are propagating in a social network. Their strategy is to reduce the positive influence near the steady vertices and increase the influence in the vacillating region \cite{LeeSMH19}.
Tong et al. propose the rumor blocking problem to find $k$ seed users to spread the truth such that the user set influenced by the rumor is minimized \cite{seed-positive}.
Some works consider more factors into the propagation models, e.g., user experience~\cite{experience-min1}, evolution of user opinions~\cite{SaxenaHLCNT20}.

In this paper, we first focus on the efficient computation of the fundamental IMIN problem under the IC model, without sacrificing the effectiveness compared with the state-of-the-art. We also further improve the quality of results by proposing a new heuristic for choosing the blockers.


\vspace{1mm}
\noindent \textbf{Influence Expected Spread Computation.}
The computation of expected spread is proved to be \#P-hard under IC model \cite{maximization1}. 
Maehara et al. propose the first algorithm to compute influence spread exactly under the IC model \cite{MaeharaSI17}, but it can only be used in small graphs with a few hundred edges. 
Domingos et al. first propose to use the Monte-Carlo Simulations (MCS) to compute the expected spread \cite{first-max}, which repeats simulations until a tight estimation is obtained.
We have to repeatedly run MCS to compute the decrease of influence spread for each candidate blocker, which leads to a large computation cost. 
Borgs et al. propose Reverse Influence Sampling (RIS) \cite{Borgs-max}, which is now widely used in IMAX Problem. Tang et al. then propose the methods to reduce the number of samples for RIS \cite{max-sota}. However, as in our later discussion, we find that RIS is not applicable to our problem (Section~\ref{sec:ec-exist}). 
Our one-time computation of the expected spread on sampled graphs can return the spread decrease of every candidate blocker, which avoids redundant computations compared with MCS.

\section{Preliminaries}
\label{sec:pre}

We consider a directed graph $G=(V,E)$, where $V$ is the set of $n$ vertices (entities), and $E$ is the set of directed edges (influence relations between vertex pairs). Table~\ref{tab:nota} summarizes the notations. When the context is clear, we may simplify the notations by omitting $V$ or $E$.

\begin{table}
    \caption{Summary of Notations}
    \small
    \label{tab:nota}
       \centering
       \resizebox{\linewidth}{!}{
       \begin{tabular}{|p{0.18\columnwidth}|p{0.74\columnwidth}|}\hline
       \rowcolor[gray]{0.95} Notation & Definition\\ \hline\hline
       $G=(V,E)$ & a directed graph with vertex set $V$ and edge set $E$\\ \hline
       $n; m$ & number of vertices/edges in $G$ (assume $m > n$)\\ \hline
       $V(G); E(G)$ & the set of vertices/edges in $G$\\ \hline
       $G[V']$ & the subgraph in $G$ induced by vertex set $V'$\\ \hline
       $N^{in}_u; N^{out}_u$ & the set of in-neighbors/out-neighbors of vertex $u$ \\ \hline
       $d^{in}_u; d^{out}_u$ & the in-degree/out-degree of vertex $u$ \\ \hline
       $S; s$ & the seed set; a seed vertex \\ \hline
       $B$ & the blocker set\\ \hline
       $\theta$ & the number of sampled graphs used in algorithm \\ \hline
       $Pr[x]$ & the probability if $x$ is true\\ \hline
       $\mathbb{E}[x]$ & the expectation of variable $x$\\ \hline
       $p_{u,v}$ & the probability that vertex $u$ activates vertex $v$\\ \hline
       $\mathcal{P}^G(x,S)$ & the probability that vertex $x$ is activated by set $S$ in $G$\\ \hline
       $\mathbb{E}(S,G)$ & the expected spread, i.e., the expected number of activated non-seed vertices in $G$ with seed set $S$ \\ \hline
       \end{tabular}}
\end{table}

\subsection{Diffusion Model}

\label{sec:diffmodel}
Following the existing studies~\cite{min-greedy1,max-rr,min-greedy2} on influence minimization, we focus on the widely-studied independent cascade (IC) model \cite{first-max}. 
It assumes each directed edge $(u,v)$ in the graph $G$ has a propagation probability\footnote{Some existing works can assign or predict the propagation probability of each edge for the IC model, e.g., \cite{first-max,maximization1,predictedge}.} $p_{u,v}\in [0,1]$, i.e., the probability that the vertex $u$ activates the vertex $v$ after $u$ is activated. 

In the IC model, each vertex has two states: inactive or active. We say a vertex is activated if it becomes active by the influence spread.
The model considers an influence propagation process as follows: (i) at timestamp $0$, the seed vertices are activated, i.e., the seeds are now active while the other vertices are inactive; (ii) if a vertex $u$ is activated at timestamp $i$, then for each of its inactive out-neighbor $v$ (i.e., for each inactive $v\in N^{out}_u$), $u$ has $p_{u,v}$ probability to independently activate $v$ at timestamp $i+1$; (iii) an active vertex will not be inactivated during the process; and (iv) we repeat the above 
steps until no vertex can be activated at the latest timestamp.


\subsection{Problem Definition}
\label{sec:define}
To formally introduce the Influence Minimization problem \cite{min-greedy1,min-greedy-tree,YaoSZWG15}, we first define the activation probability of a vertex, which is initialized by $1$ for any seed by default.

\begin{definition}[activation probability]
Given a directed graph $G$, a vertex $x$ and a seed set $S$, the activation probability of $x$ in $G$, denoted by $\mathcal{P}^G(x,S)$, is the probability of the vertex $x$ becoming active. 
\end{definition}

In order to minimize the spread of misinformation, we can block some key non-seed vertices such that they will not be activated in the propagation process.
A blocked vertex is also called a blocker in this paper.

\begin{definition}[blocker]
\revise{Given $G=(V,E)$, a seed set $S$, and a set of blockers $B\subseteq (V\setminus S)$, the influence probability of every edge pointing to a vertex in $B$ is set to $0$, i.e., $p_{u,v} = 0$ for any $u\in N_v^{in}$ if $v\in B$.}
\end{definition}

\revise{The activation probability of a blocker is $0$ because the propagation probability is $0$ for any of its incoming edges.} Then, we define the expected spread to measure the influence of the seed set in the whole graph.

\begin{definition}[expected spread]
Given a directed graph $G$ and a seed set $S$, the expected spread, denoted by $\mathbb{E}(S,G)$, is the expected number of active vertices under the IC model, i.e., $\mathbb{E}(S,G)=\sum_{u\in V(G)}\mathcal{P}^G(u,S)$.
\end{definition}


The expected spread with a blocker set $B$ is represented by $\mathbb{E}(S,G[V\setminus B])$.
The studied problem is defined as follows.



\vspace{2mm} \noindent 
{\bf Problem Statement.} 
Given a directed graph $G=(V,E)$, the influence probability $p_{u,v}$ on each edge $(u,v)$, a seed set $S$ and a budget $b$, the Influence Minimization (IMIN) problem is to find a blocker set $B^*$ with at most $b$ vertices such that the influence (i.e., expected spread) is minimized, i.e., 
$$B^*={\arg \min}_{B\subseteq (V\setminus S),|B|\le b} \mathbb{E}(S,G[V\setminus B]).$$


\begin{figure}[t]
    \centering
        \centering    
        \includegraphics[width=.85\columnwidth]{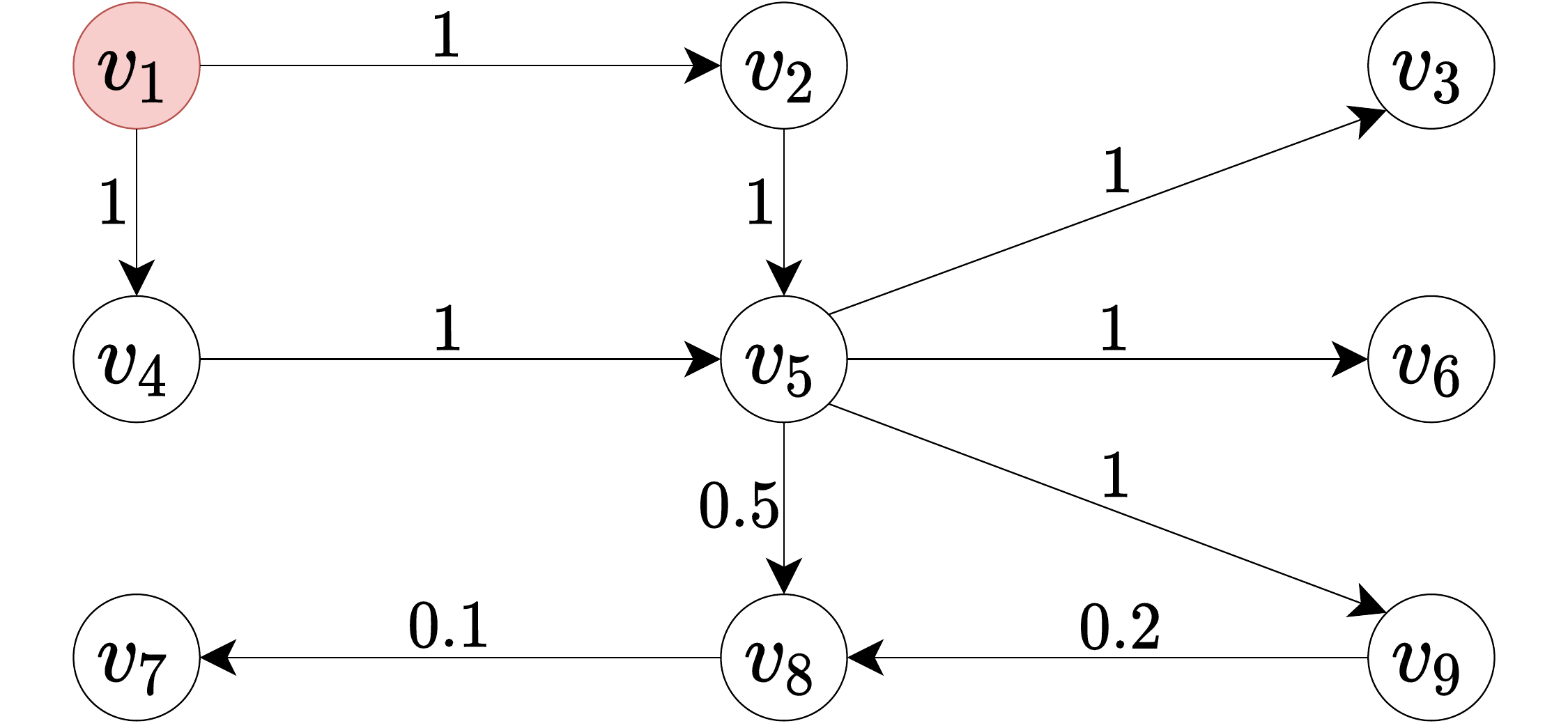}
        
    \caption{A toy graph $G$, where $v_1$ is the seed vertex and the value on each edge indicates its propagation probability.}
\label{fig:sample}
\end{figure}

\begin{example}
\label{example:compute}
Figure~\ref{fig:sample} shows a graph $G = (V, E)$ where $S=\{v_1\}$ is the seed set, and the value on each edge is its propagation probability, e.g., $p_{v_5,v_8}=0.5$ indicates $v_8$ can be activated by $v_5$ with $0.5$ probability if $v_5$ becomes active. 
At timestamps $1$ to $3$, the seed $v_1$ will certainly activate $v_2,v_3,v_4,v_5,v_6$ and $v_9$, as the corresponding activation probability is $1$. 
Because $v_8$ may be activated by either $v_5$ or $v_9$, we have $\mathcal{P}^G(v_8,\{v_1\})=1-(1-p_{v_5,v_8}\cdot\mathcal{P}^G(v_5,\{v_1\}))(1-p_{v_9,v_8}\cdot\mathcal{P}^G(v_9,\{v_1\}))=1-0.5\times 0.8 = 0.6$. 
If $v_8$ is activated, it has $0.1$ probability to activate $v_7$.
Thus, we have $\mathcal{P}^G(v_7,\{v_1\})=p_{v_8,v_7}\cdot\mathcal{P}^G(v_8,\{v_1\})=0.6\times 0.1=0.06$. 
The expected spread is the activation probability sum of all the vertices, i.e., $\mathbb{E}(\{v_1\},G)=7.66$. 
If we block $v_5$, the new expected spread $\mathbb{E}(\{v_1\},G[V\setminus \{v_5\}])=3$. Similarly, we have $\mathbb{E}(\{v_1\},G[V\setminus \{v_2\}])=\mathbb{E}(\{v_1\},G[V\setminus \{v_4\}])=6.66$, and blocking any other vertex also achieves a smaller expected spread than blocking $v_5$. Thus, if the budget $b = 1$, the result of the IMIN problem is $\{v_5\}$.
\end{example}

\section{Problem Analysis}
\label{sec:analysis}

To the best of our knowledge, no existing work has studied the hardness of the Influence Minimization (IMIN) problem, as surveyed in~\cite{DBLP:journals/jnca/ZareieS21}. Thus, we first analyze the problem hardness.

\begin{theorem}
The IMIN problem is NP-Hard.
\label{theo::NP-IM}
\end{theorem}

\begin{proof}
We reduce the densest k-subgraph (DKS) problem~\cite{BhaskaraCCFV10}, which is NP-hard, to the IMIN problem. Given an undirected graph $G(V,E)$ with $|V|=n$ and $|E|=m$, and an positive integer $k$, the DKS problem is to find a subset $A\subseteq V$ with exactly $k$ vertices such that the number of edges induced by $A$ is maximized.

Consider an arbitrary instance $G(V=\{v_1,\cdots, v_n\},E=\{e_1,\cdots, e_{m}\})$ of DKS problem with a positive integer $k$, we construct a corresponding instance of IMIN problem on graph $G'$. Figure~\ref{fig:np-hard} shows a construction example from $4$ vertices and $4$ edges.

The graph $G'$ contains three parts: $C$, $D$, and a seed vertex $S$. The part $C$ contains $n$ vertices, i.e., $C=\cup_{1\le i\le n} c_i$ where each $c_i$ corresponds to $v_i$ of instance $G$. The part $D$ contains $m$ vertices, i.e., $D=\cup_{1\le i\le m} d_i$ where each vertex $d_i$ corresponds to edge $e_i$ in instance $G$. The vertex $S$ is the only seed of the graph.
For each edge $e_i=(v_x,v_y)$ in graph $G$, we add two edges: (i) from $c_x$ to $d_i$ and (ii) $c_y$ to $d_i$. Then we add an edge from $S$ to each $c_i (1\le i\le n)$. The propagation probability of each edge is set to $1$.
The construction of $G'$ is completed.

We then show the equivalence between the two instances.
As the propagation probability of each edge is $1$, the expected spread in the graph is equal to the number of vertices that can be reached by seed $S$.
Adding a vertex $v_i$ into vertex set $A$ corresponds to the removal of $c_i$ from the graph $G'$, i.e., $c_i$ is blocked. If edge $e_i$ is in the induced subgraph, the corresponding vertex $d_i$ in $G'$ cannot be reached by seed $S$. We find that blocking the vertices $c_i\in C$ will first lead to the decrease of expected spread of themselves, and the vertices in $D$ may also not be reached by $S$ if both two in-neighbors of them are blocked. Thus, blocking the corresponding vertices $c_i$ of $v_i\in A$ will lead to $|A|+g$ decrease of expected spread, where $g$ is the number of vertices in $D$ that cannot be reached by $S$ (equals to the number of edges in the induced subgraph $G[A]$). Note that there is no need to block vertices in $D$, because they do not have any out-neighbors, and blocking them only leads to the decrease of expected spread of themselves which is not larger than the decrease of the expected spread of blocking the vertices in $C$. 
{We can find that IMIN problem will always block $k$ vertices, as blocking one vertex will lead to at least $1$ decrease of expected spread.}

The optimal blocker set $B$ of IMIN problem corresponds to the optimal vertex set $A$ of DKS problem, where each vertex $c_i\in B$ corresponds to the set $v_i\in A$.
Thus, if the IMIN problem can be solved in PTIME, the DKS problem can also be solved in PTIME, while the DKS problem is NP-hard.
\end{proof}

\begin{figure}
    \centering
    \includegraphics[width=1\linewidth]{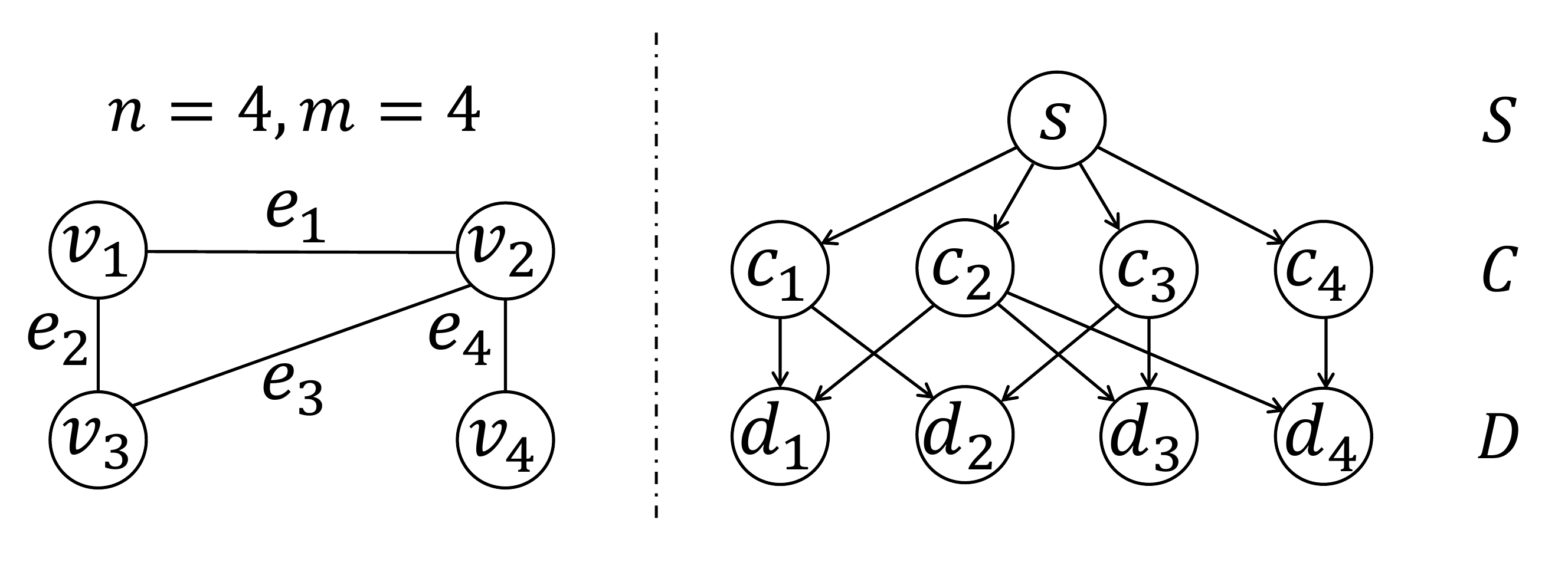}
    \caption{Construction example for hardness proofs, where the left part of the figure is an instance of DKS problem and the right part is its corresponding construction of IMIN problem.}
    \label{fig:np-hard}
\end{figure}

We also show that the function of expected spread is monotone and not supermodular under IC model.

\begin{theorem}
Given a graph $G$ and a seed set $S$, the expected spread function $\mathbb{E}(S,G[V \setminus B])$ is monotone and not supermodular of $B$ under IC model.
\label{theo::expectedspread}
\end{theorem}

\begin{proof}
As adding any blocker to any set $B$ cannot increase the expected influence spread, we have $\mathbb{E}(S,G[V \setminus B])$ is monotone of $B$.
For every two set $X\subseteq Y\subseteq V$ and vertex $x\in V\setminus Y$, if function $f(\cdot)$ is supermodular, it must hold that $f(X\cup \{x\})-f(X)\le f(Y\cup \{x\})-f(Y)$. 
Consider the graph in Figure~\ref{fig:sample}, let $f(B)=\mathbb{E}(S,G[V \setminus B])$, $X=\{v_3\},Y=\{v_2,v_3\}$ and $x=v_4$. As $f(X)=6.66$, $f(Y)=5.66$, $f(X\cup \{x\})=5.66$, and $f(Y\cup \{x\})=1$, we have $f(X\cup \{x\})-f(X)=-1 > f(Y\cup \{x\})-f(Y)=-4.66$.
\end{proof}

In addition, we prove that the IMIN problem under IC model is hard to approximate.

\begin{theorem}
Under IC model, the IMIN problem is APX-hard unless P=NP.
\label{theo::APX}
\end{theorem}

\begin{proof}
We use the same reduction from the densest k-subgraph (DKS) problem to the Influence Minimization problem, as in the proof of Theorem~\ref{theo::NP-IM}.
Densest k-subgraph (DKS) problem does not have any polynomial-time approximation scheme, unless P=NP~\cite{Khot06}.
According to the proof of Theorem~\ref{theo::NP-IM}, we have a blocker set $B$ for influence minimization problem on $G'$ corresponding to a vertex set $A$ for DKS problem, where each $c_i\in B$ corresponds to $v_i\in A$.
Let $x$ denote the number of edges in the optimal result of DKS, and $y$ denote the optimal spread in IMIN, we have $x+k=y$, where $k$ is the given positive number of DKS problem.
If there is a solution with $\gamma$-approximation on the influence minimization problem, there will be a $\lambda$-approximation on the DKS problem.
Thus, there is no PTAS for the influence minimization problem, and it is APX-hard unless P=NP.
\end{proof}
\section{Existing works and our approach
}
\label{sec:appro}

The hardness of the problem motivates us to develop an effective and efficient heuristic algorithm. 
In this section, we first introduce the state-of-the-art solution (i.e., the greedy algorithm with Monte-Carlo Simulations) as the baseline algorithm (Section~\ref{sec:baseline}). Then, we analyze existing solutions for expected spread computation, and propose a new estimation algorithm based on sampled graphs and dominator trees to compute the decrease of expected spread for all vertices at once (Section~\ref{sec:candidates}). Applying the new framework of expected spread estimation for selecting the candidates, we propose our AdvancedGreedy algorithm (Section~\ref{sec:advanced-greedy}) with higher efficiency and without sacrificing the effectiveness, compared with the baseline.
As the greedy approaches do not consider the cooperation of candidate blockers during the selection, some important vertices may be missed, e.g., some out-neighbors of the seed.
Thus, we further propose a superior heuristic, the GreedyReplace algorithm, to achieve a better result quality (Section~\ref{sec:greedyreplace}).


\label{sec:onesource}
\vspace{1.5mm}
\noindent
\textbf{From Multiple Seeds to One Seed.}
For presentation simplicity, we introduce the techniques for the case of one seed vertex.
A unified seed vertex $s'$ is created to replace all the seeds in the graph. For each vertex $u$, if there are $h$ different seeds pointing to $u$ and the probability on each edge is $p_i (1\le i\le h)$, we remove all the edges from the seeds to $u$ and add an edge from $s'$ to $u$ with probability $(1-\prod_{i=1}^h (1-p_i))$. 
As an active vertex in the IC model only has one chance to activate every out-neighbor, the above modification will not affect the influence spread (i.e., expected spread in the graph) and the resulting blocker set is the same as the original problem.

\subsection{Baseline Algorithm}
\label{sec:baseline}

We first review and discuss the baseline greedy algorithm, which is the state-of-the-art for influence minimization (IMIN) problem and its variants \cite{min-greedy1,min-greedy-tree,min-greedy2,acycle-min-greedy,PhamPHPT19}. 

The greedy algorithm for the IC model is as follows: we start with an empty blocker set $B=\varnothing$, and then iteratively add vertex $u$ into set $B$ that leads to the largest decrease of expected spread, i.e., $u=\arg \max _{u\in V\setminus (S\cup B)}(\mathbb{E}(S,G[V\setminus B])-\mathbb{E}(S,G[V\setminus (B\cup \{u\})])$, until $|B|=b$.

Algorithm~\ref{algo:greedy} shows the details of the baseline greedy algorithm. The algorithm starts from an empty blocker set $B$ (Line 1). Then, in each iteration (Line 2), $x$ records the vertex whose blocking corresponds to the largest decrease of expected spread  (Line 3). The baseline greedy algorithm enumerates all the vertices to find the blocker with the maximum decrease of expected spread in each round (Lines 4-7) and insert it into the blocker set (Line 8). After $b$ iterations, the algorithm returns the blocker set $B$ (Line 9). 

\begin{algorithm}[t]
    \SetVline 
    \SetFuncSty{textsf}
    \SetArgSty{textsf}
	\caption{BaselineGreedy($G,s$)}
	\label{algo:greedy}
	\Input{a graph $G$ and the source $s$}
	\Output{the blocker set $B$}
	\State{$B\leftarrow $ empty}
	\For{$i\leftarrow 1$ to $b$}
	{
	    \State{$x\leftarrow -1$}
	    \For{each vertex in $u\in V(G)\setminus (B\cup \{s\})$}
	    {
	        \State{$\Delta [u]\leftarrow$ decrease of expected spread when blocking $u$}
	        \If{$x=-1$ or $\Delta [u] > \Delta[x]$}
	        {
	            \State{$x\leftarrow u$}
	        }
	    }
	    \State{$B\leftarrow B\cup \{x\}$}
	}
	\Return{$B$}
\end{algorithm}

As the previous works use Monte-Carlo Simulations to compute the expected spread for the greedy algorithm (Line 5 in Algorithm~\ref{algo:greedy}), each computation of spread decrease needs $O(r\cdot m)$ time, where $r$ is the number of rounds in Monte-Carlo Simulations. Thus, the time complexity of Algorithm~\ref{algo:greedy} is $O(b\cdot n\cdot r\cdot m)$.

As indicated by the complexity, the baseline greedy algorithm cannot efficiently handle the cases with large $b$. 
The greedy heuristic is usually effective on small $b$ values, while the time cost is still large because it has to enumerate the whole vertex set as the candidate blockers and compute the expected spread for each candidate. 


\subsection{Efficient Algorithm for Candidate Selection}
\label{sec:candidates}

In this subsection, we propose an efficient algorithm for selecting the candidates.
We first show that the existing solutions for computing expected spread are infeasible to solve the IMIN problem efficiently (Section~\ref{sec:ec-exist}). Then, we propose a new framework (Section~\ref{sec:new-es}) based on sampled graphs (Section~\ref{sec::sample-graph}) and their dominator trees (Section~\ref{sec::domi-tree}) which can quickly compute the decrease of the expected spread of every candidate blocker through only one scan on the dominator trees.

\vspace{2.5mm}
\subsubsection{Existing Works}
\label{sec:ec-exist}

As computing the expected influence spread of a seed set in IC model is \#P-hard  \cite{maximization1}, and the exact solution can only be used in small graphs (e.g., with a few hundred edges) \cite{MaeharaSI17}. Thus, the existing works focus on estimation algorithms.
There are two directions as follows.

\vspace{2mm}
\noindent \textbf{Monte-Carlo Simulations (MCS).}
Kempe et al. \cite{first-max} apply Monte-Carlo Simulations to estimate the influence spread under IC model, which is often used in some influence related problems, e.g., \cite{OhsakaAYK14,MS-PCG,influencepath}. In each round of MCS, it removes every edge $(u,v)$ with $(1-p_{u,v})$ probability. Let $G'$ be the resulting graph, and the set $R(s)$ contains the vertices in $G'$ that are reachable from $s$ (i.e., there exists at least one path from $s$ to each vertex in $R(s)$). For the original graph $G$ and seed $s$, the expected size of set $R(s)$ equals to the expected spread $\mathbb{E}(\{s\},G)$ \cite{first-max}. 
Assuming we take $r$ rounds of MCS to estimate the expected spread, MCS needs $O(r\cdot m)$ times to calculate the expected spread. 
Recall that the influence minimization problem is to find the optimal blocker set with a given seed set.
The spread computation by MCS for influence minimization is costly, because the dynamic of influence spread caused by different blockers is not fully utilized in the sampling, and we have to repeatedly conduct MCS for each candidate blocker set.

\vspace{2mm}
\noindent \textbf{Reverse Influence Sampling (RIS).}
Borgs et al. \cite{Borgs-max} propose the Reverse Influence Sampling to approximately estimate the influence spread,  which is often used in the solutions for Influence Maximization (IMAX) problem, e.g., \cite{SunHYC18,Guo0WC20}. For each round, RIS generates an instance of $g$ randomly sampled from graph $G$ by removing each edge $(u,v)$ in $G$ with $(1-p_{u,v})$ probability, and then randomly sample a vertex $x$ in $g$. It performs reverse BFS to compute the reverse reachable (RR) set of the vertex $x$, i.e., the vertices which can be reached by vertex $x$ in the reverse graph of $g$. They prove that if the RR set of vertex $v$ has $\rho$ probability to contain the vertex $u$, when $u$ is the seed vertex, we have $\rho$ probability to activate $v$. 
In the IMAX problem, RIS generates RR sets by sampling the vertices in the sampled graphs and then applying the greedy heuristic. 
As the expected influence spread is submodular of seed set $S$ \cite{expectedspread}, an approximation ratio can be guaranteed by RIS in IMAX problem.
However, for our problem, reversing the graph is not helpful as the blockers seem ``intermediary" between the seeds and other vertices s.t. the computation cannot be unified into a single process in the reversing.
We prove the expected spread is not supermodular of blocker set $B$ which implies the absolute value of the marginal gain does not show a diminishing return.
Thus, the decrease of expected spread led by a blocker combination cannot be determined by the union effect of single blockers in the combination.



\begin{figure*}[t]
\centering
\begin{subfigure}[h]{0.205\linewidth}
    \includegraphics[width=1\columnwidth]{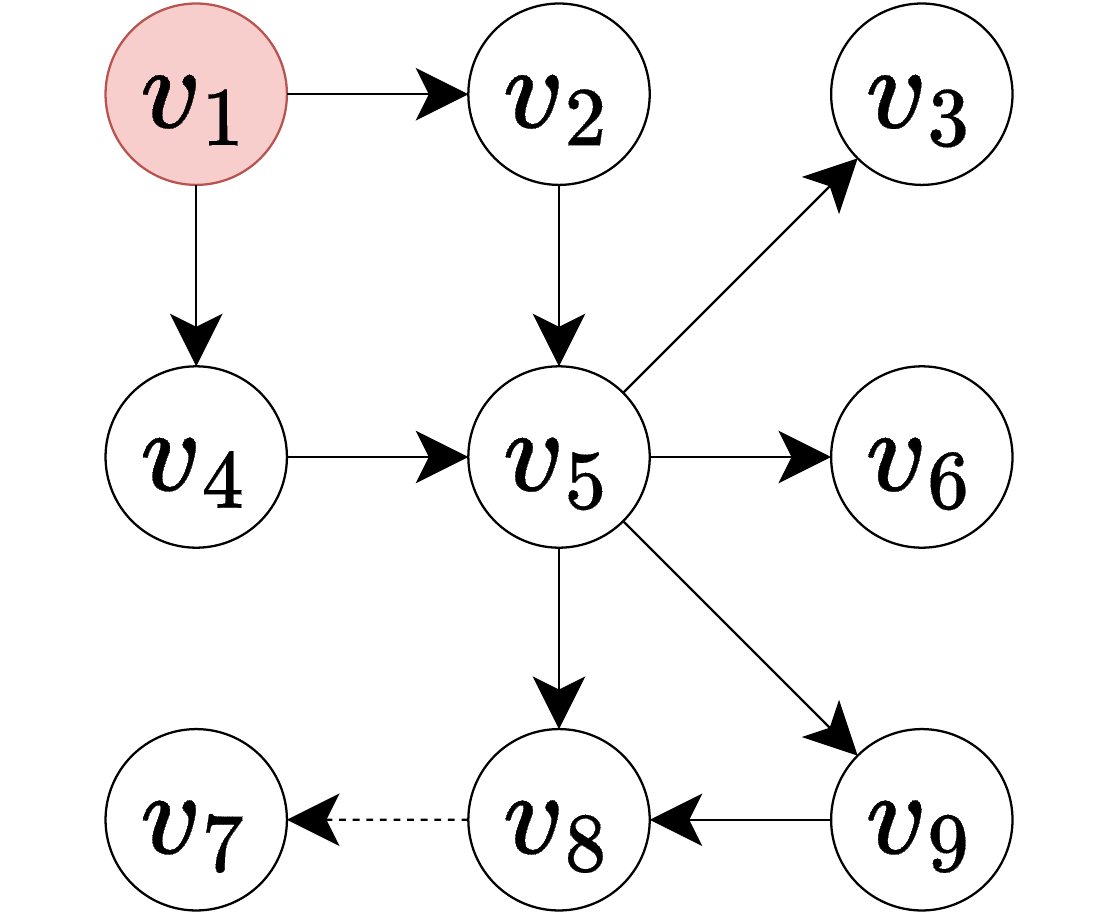}
    \caption{Sampled graph 1}\label{fig:sample-example-a}
\end{subfigure}
\begin{subfigure}[h]{0.205\linewidth}
    \includegraphics[width=1\columnwidth]{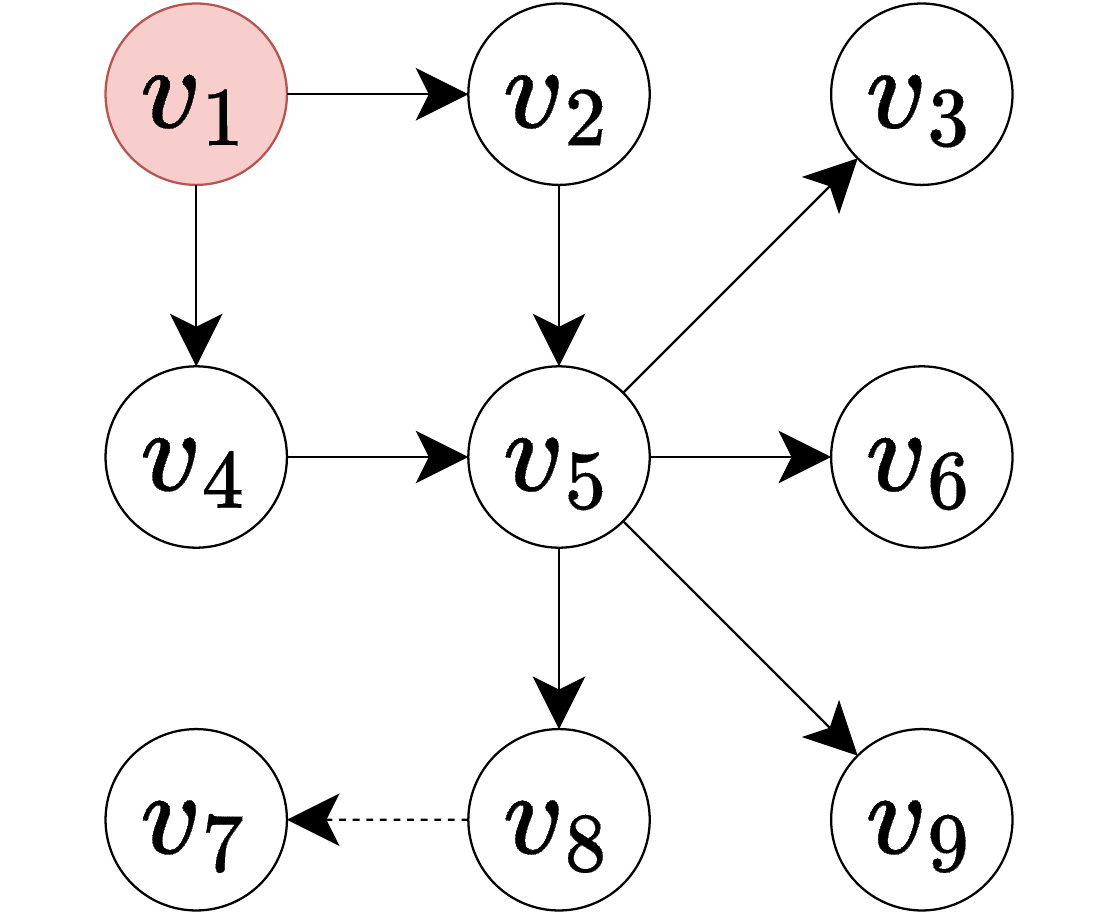}
    \caption{Sampled graph 2}\label{fig:sample-example-b}
\end{subfigure}
\begin{subfigure}[h]{0.205\linewidth}
    \includegraphics[width=1\columnwidth]{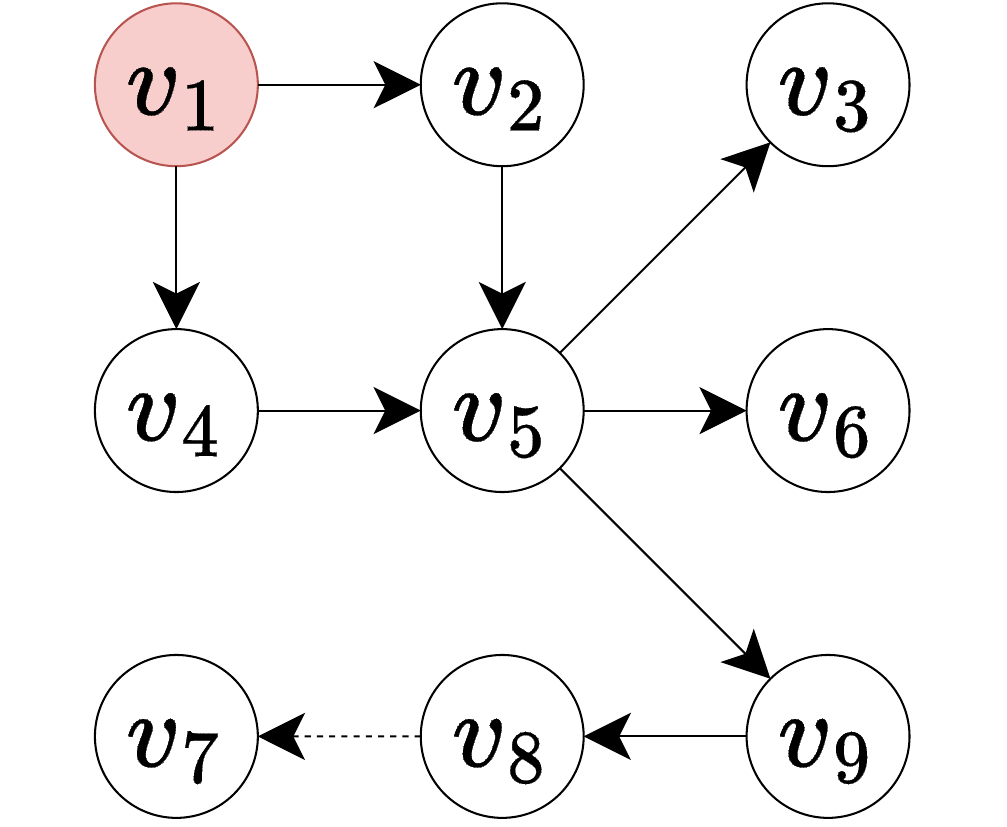}
    \caption{Sampled graph 3}\label{fig:sample-example-c}
\end{subfigure}
\begin{subfigure}[h]{0.205\linewidth}
    \includegraphics[width=1\columnwidth]{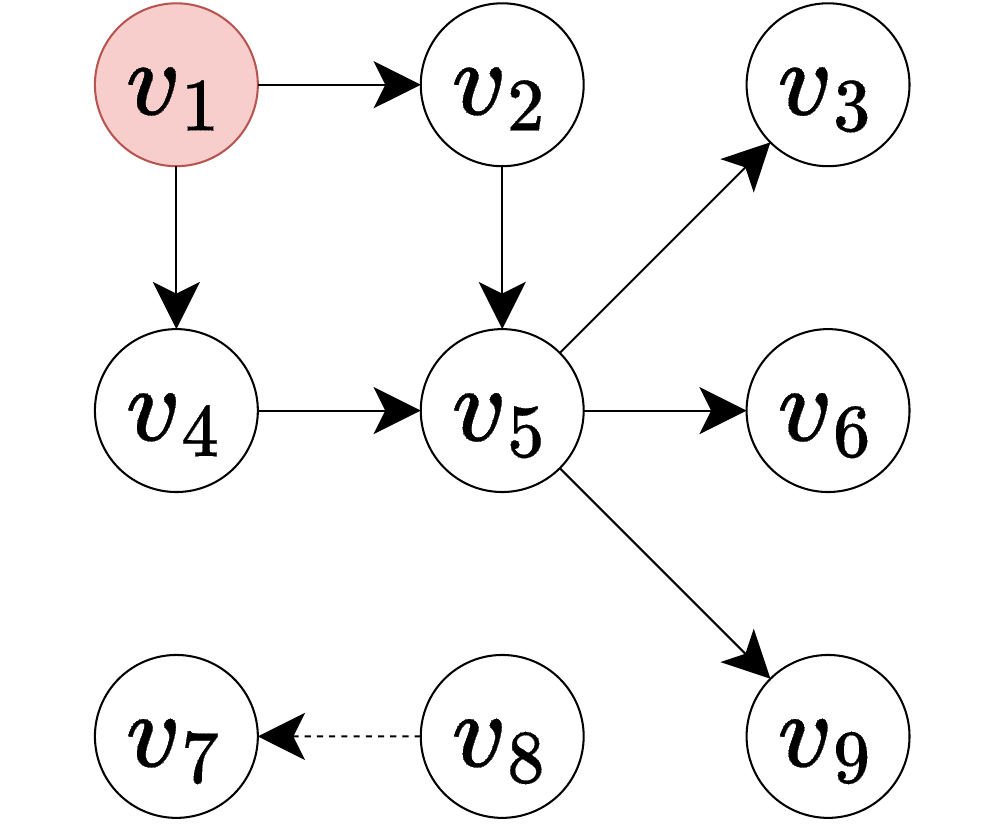}
    \caption{Sampled Graph 4}\label{fig:sample-example-d}
\end{subfigure}
\caption{Sampled graphs of the graph $G$ in Figure~\ref{fig:sample}.}
\label{fig:sample-example}
\end{figure*}

\begin{figure*}[t]
\centering
\begin{subfigure}[h]{0.205\linewidth}
    \includegraphics[width=1\columnwidth]{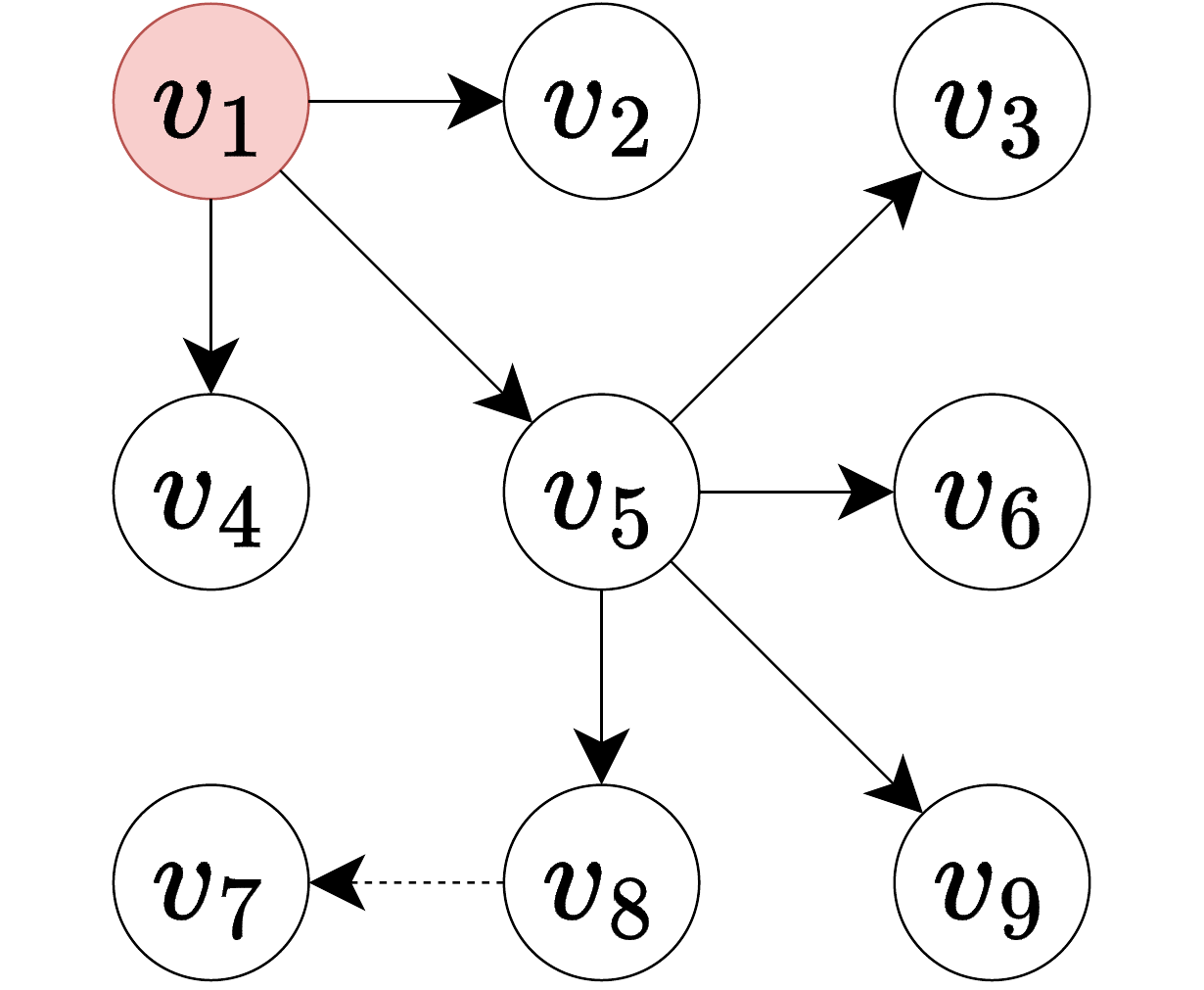}
    \caption{DT of Sampled Graph 1}\label{fig:sample-dt-example-a}
\end{subfigure}
\begin{subfigure}[h]{0.205\linewidth}
    \includegraphics[width=1\columnwidth]{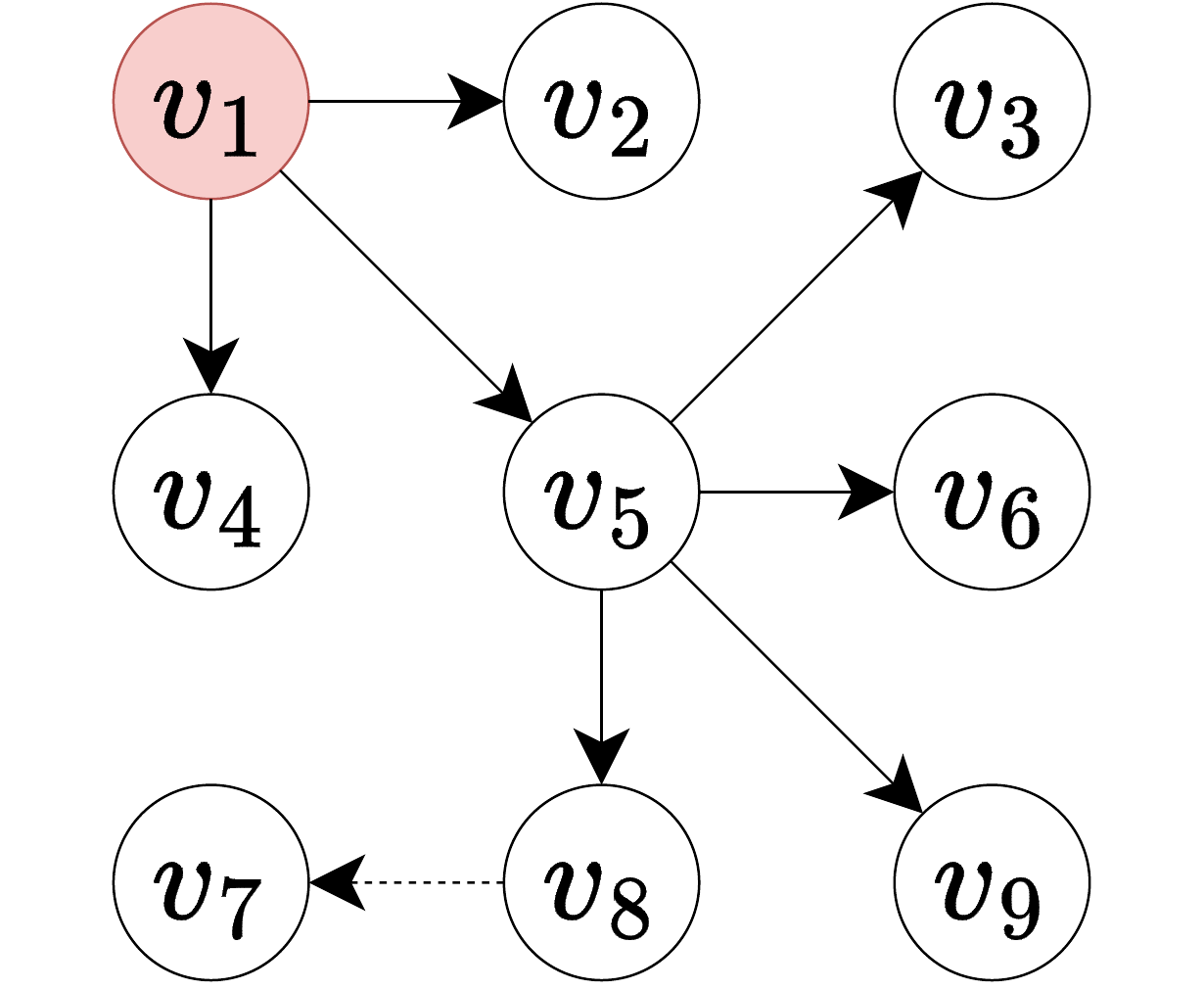}
    \caption{DT of Sampled Graph 2}\label{fig:sample-dt-example-b}
\end{subfigure}
\begin{subfigure}[h]{0.205\linewidth}
    \includegraphics[width=1\columnwidth]{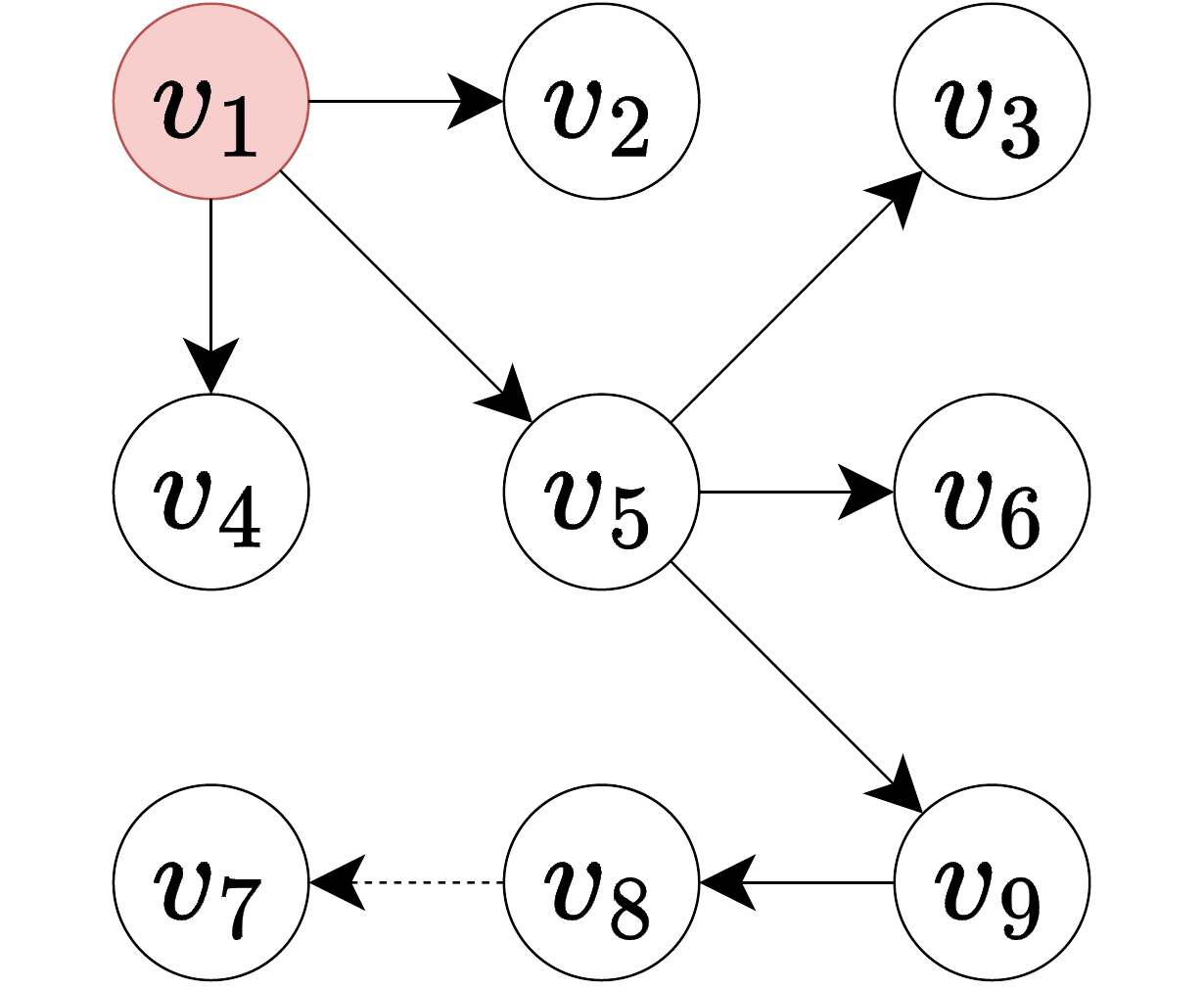}
    \caption{DT of Sampled Graph 3}\label{fig:sample-dt-example-c}
\end{subfigure}
\begin{subfigure}[h]{0.205\linewidth}
    \includegraphics[width=1\columnwidth]{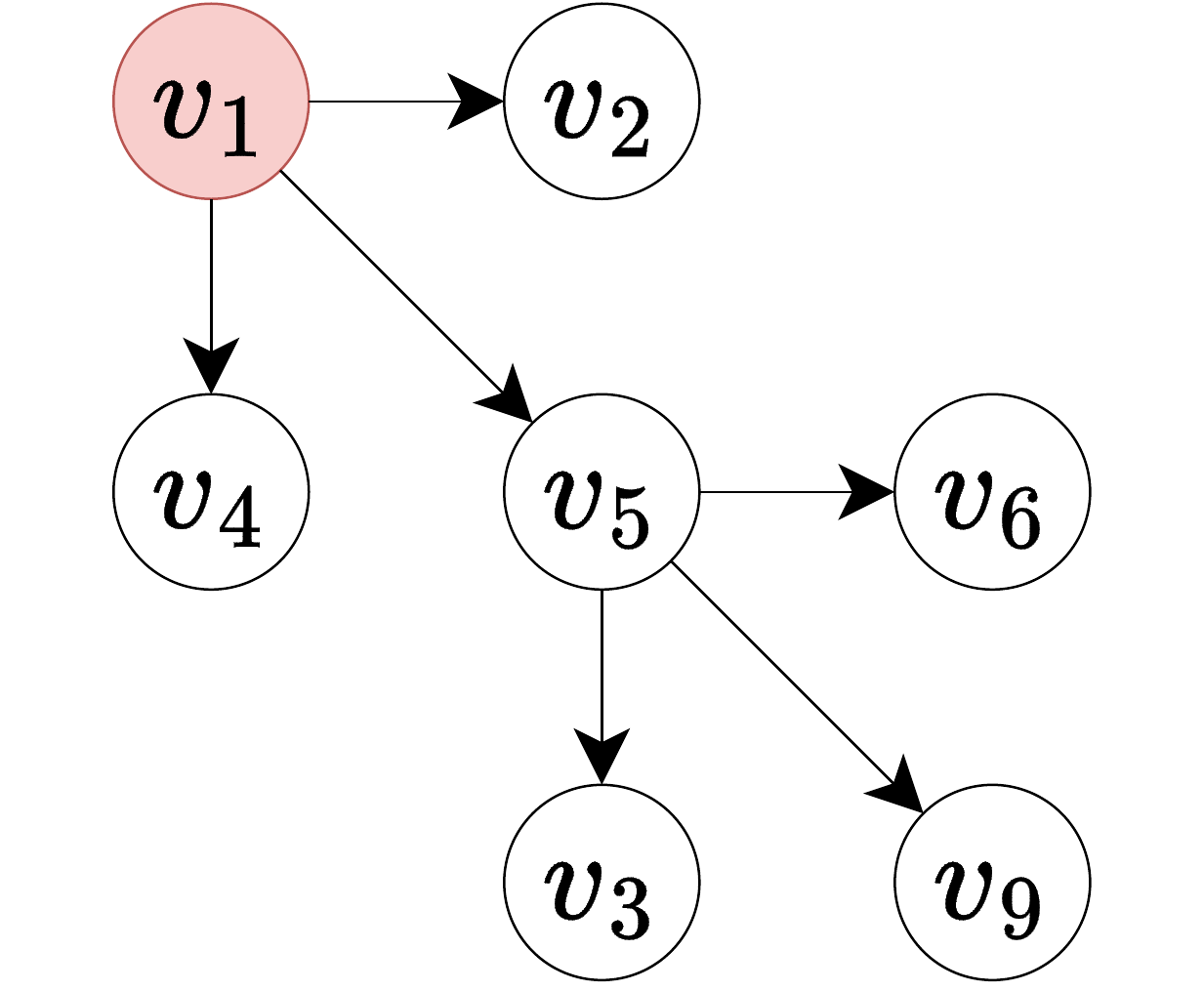}
    \caption{DT of Sampled Graph 4}\label{fig:sample-dt-example-d}
\end{subfigure}
\caption{Dominator trees of sampled graphs in Figure~\ref{fig:sample-example}.}
\label{fig:sample-dt-example}
\end{figure*}

\vspace{2mm}
\subsubsection{Estimate the Expected Spread Decrease with Sampled Graphs} \label{sec::sample-graph}
We first define the random sampled graph.

\begin{definition}[Random Sampled Graph]
Let $\mathcal{G}$ be the distribution of the graphs with each induced by the randomness in edge removals from $G$, i.e., removing each edge $e=(u,v)$ in $G$ with $1-p_{u,v}$ probability. A random sampled graph $g$ derived from $G$ is an instance randomly sampled from $\mathcal{G}$.
\end{definition}

\begin{table}
    \caption{\revise{Summary of Notations for Random Sampled Graph}}
    \small
    \label{tab:nota-sampled}
       \centering
       \resizebox{\linewidth}{!}{
       \begin{tabular}{|p{0.15\columnwidth}|p{0.75\columnwidth}|}\hline
       \rowcolor[gray]{0.95} \revise{Notation} & \revise{Definition}\\ \hline\hline
       \revise{$\sigma(s,G)$} &
       \revise{the number of vertices reachable from $s$ in $G$}\\ \hline
       \revise{$\sigma^{\rightarrow u}(s,G)$} &
\revise{the number of vertices reachable from $s$ in $G$, where all the paths from $s$ to these vertices pass through $u$}
       \\ \hline
       \revise{$\xi^{\rightarrow u}(s,G)$} & \revise{the average number of $\sigma^{\rightarrow u}(s,G)$ in the sampled graphs which are derived from $G$}\\ \hline
       \end{tabular}}
\end{table}

\revise{We summarize the notations related to the random sampled graph in Table~\ref{tab:nota-sampled}.}
The following lemma is a useful interpretation of expected spread with IC model~\cite{MaeharaSI17}.

\begin{lemma}
\label{lemma:sample}
Suppose that the graph $g$ is a random sampled graph derived from $G$. Let $s$ be a seed vertex, we have $\mathbb{E}[\sigma(s,g)]=\mathbb{E}(\{s\},G)$.
\end{lemma}

\revise{By Lemma~\ref{lemma:sample}, we have the following corollary for computing the expected spread when blocking one vertex.}

\begin{corollary}
\label{corollary::es-onevertex}
\revise{Given two fixed vertices $s$ and $u$ with $s\neq u$, and a random sampled graph $g$ derived from $G$, we have $\mathbb{E}[\sigma(s,g)-\sigma^{\rightarrow u}(s,g)]=\mathbb{E}(\{s\},G[V\setminus \{u\}])$.}
\end{corollary}

\begin{proof}
\revise{Let $g'=g[V(g)\setminus \{u\}]$, we have $\sigma(s,g)-\sigma^{\rightarrow u}(s,g)=\sigma(s,g')$. Thus, $\mathbb{E}[\sigma(s,g)-\sigma^{\rightarrow u}(s,g)]=\mathbb{E}[\sigma(s,g')]$. As $g'=g[V(g)\setminus \{u\}]$, $g'$ can be regarded as a random sampled graph derived from graph $G[V\setminus \{u\}]$. By Lemma~\ref{lemma:sample}, $\mathbb{E}[\sigma(s,g')]=\mathbb{E}(\{s\},G[V\setminus \{u\}])$. Therefore, $\mathbb{E}[\sigma(s,g)-\sigma^{\rightarrow u}(s,g)]=\mathbb{E}(\{s\},G[V\setminus \{u\}])$.}
\end{proof}

\revise{Based on Lemma~\ref{lemma:sample} and Corollary~\ref{corollary::es-onevertex}, we can compute the decrease of expected spread when a vertex is blocked.}

\begin{theorem}
Let $s$ be a fixed vertex, $u$ be a blocked vertex, and $g$ be a random sampled graph derived from $G$, respectively. For any vertex $u\in V(G)\setminus \{s\}$, we have the decrease of expected spread by blocking $u$ is equal to $\mathbb{E}[\sigma^{\rightarrow u}(s,g)]$, where $\mathbb{E}[\sigma^{\rightarrow u}(s,g)]=\mathbb{E}(\{s\},G)-\mathbb{E}(\{s\},G[V\setminus \{u\}])$.
\label{theo::block_expected}
\end{theorem}

\begin{proof}
\revise{
As $\sigma^{\rightarrow u}(s,g)=\sigma(s,g)-(\sigma(s,g)-\sigma^{\rightarrow u}(s,g))$, we have $\mathbb{E}[\sigma^{\rightarrow u}(s,g)]=\mathbb{E}[\sigma(s,g)]-\mathbb{E}[\sigma(s,g)-\sigma^{\rightarrow u}(s,g)]$. 
From Corollary~\ref{corollary::es-onevertex}, we have $\mathbb{E}[\sigma(s,g)-\sigma^{\rightarrow u}(s,g)]=\mathbb{E}(\{s\},G[V\setminus \{u\}])$.
By Lemma~\ref{lemma:sample}, we know $\mathbb{E}[\sigma(s,g)]=\mathbb{E}(\{s\},G)$.
Therefore, $\mathbb{E}[\sigma^{\rightarrow u}(s,g)]=\mathbb{E}(\{s\},G)-\mathbb{E}(\{s\},G[V\setminus \{u\}])$.}
\end{proof}

As we use random sampling for estimating the decrease of expected spread of each vertex, we show that the average number of $\sigma^{\rightarrow u}(s,g)$ is an accurate estimator of any vertex $u$ and fixed seed vertex $s$, when the number of sampled graphs is sufficiently large. Let $\theta$ be the number of sampled graphs, $\xi^{\rightarrow u}(s,G)$ be the average number of $\sigma^{\rightarrow u}(s,g)$ and $OPT$ be the exact decrease of expected spread from blocking vertex $u$, i.e., $OPT=\mathbb{E}[\sigma^{\rightarrow u}(s,g)]$ (Theorem~\ref{theo::block_expected}). We use the Chernoff bounds~\cite{MotwaniR95} for theoretical analysis.

\begin{lemma}
Let $X$ be the sum of $c$ i.i.d. random variables sampled from a distribution on $[0,1]$ with a mean $\mu$. For any $\delta >0$, we have $Pr[X-c\mu \ge \delta \cdot c\mu] \le e^{-\delta^2c\mu /(2+\delta)}$ and $Pr[X-c\mu \le -\delta \cdot c\mu] \le e^{-\delta^2c\mu /2}.$
\label{lemma:cher}
\end{lemma}

\begin{theorem}
\label{theorem:approx}
For seed vertex $s$ and a fixed vertex $u$, the inequality $|\xi^{\rightarrow u}(s,G)-OPT| < \varepsilon\cdot OPT$ holds with at least $1-{n^{-l}}$ probability when $\theta\ge \frac{l(2+\varepsilon)n\log n}{\varepsilon^2\cdot OPT}$.
\end{theorem}

\begin{proof}
We have
$Pr[|\xi^{\rightarrow u}(s,G)-OPT|\ge \varepsilon\cdot OPT]$
$
\begin{matrix}
 =& Pr[|\frac{\xi^{\rightarrow u}(s,G)\cdot \theta}{n}-\frac{OPT\cdot \theta}{n}|\ge \frac{\varepsilon \cdot \theta}{n}\cdot OPT] \\
\end{matrix}
$.

Let $\delta=\varepsilon$, $\mu = \frac{OPT}{n}$ and $X=\frac{\xi^{\rightarrow u}(s,G)\cdot \theta}{n}$.
According to Lemma~\ref{lemma:cher}, we have
$Pr[|\xi^{\rightarrow u}(s,G)-OPT|\ge \varepsilon\cdot OPT]$
$=  Pr[|X-\theta \mu|\ge \delta \cdot \theta \mu] 
\le exp(\frac{-\delta^2\theta\mu}{2+\delta})
 = exp(\frac{-\varepsilon^2\theta OPT}{n(2+\varepsilon)})$.


As $\theta\ge \frac{l(2+\varepsilon)n\log n}{\varepsilon^2\cdot OPT}$, we have $Pr[|\xi^{\rightarrow u}(s,G)-OPT|\ge \varepsilon\cdot OPT]\le exp(l \log n) = n^{-l}$.
Therefore, $|\xi^{\rightarrow u}(s,G)-OPT| < \varepsilon\cdot OPT$ holds with at least $1-n^{-l}$ probability.
\end{proof}

\subsubsection{Dominator Trees of Sampled Graphs}
\label{sec::domi-tree}


In order to efficiently compute the decrease of expected spread when each vertex is blocked, we apply Lengauer-Tarjan algorithm to construct the dominator tree~\cite{LengauerT79}. 
Note that, in the following of this subsection, the id of each vertex is reassigned by the sequence of a DFS on the graph starting from the seed.


\begin{definition}[dominator]
Given $G=(V,E)$ and a source $s$, the vertex $u$ is a dominator of vertex $v$ when every path in $G$ from $s$ to $v$ has to go through $u$.
\end{definition}

\begin{definition}[immediate dominator]
Given $G=(V,E)$ and a source $s$, the vertex $u$ is the immediate dominator of vertex $v$, denoted $idom(v)=u$, if $u$ dominates $v$ and every other dominator of $v$ dominates $u$.
\end{definition}

We can find that every vertex except source $s$ has a unique immediate dominator. The dominator tree of graph $G$ is induced by the edge set $\{(idom(u),u)\mid u\in V\setminus \{s\}\}$ with root $s$ \cite{AhoU73,LowryM69}.

Lengauer-Tarjan algorithm proposes an efficient algorithm for constructing the dominator tree.
It first computes the semidominator of each vertex $u$, denoted by $sdom(u)$, where $sdom(u)=\min \{v\mid \text{there is a path~} v=v_0,v_1,\cdots, v_k=u \text{~with~} v_i>u \text{~for any integer~} i\in [1, k)\}$. The semidominator can be computed by finding the minimum $sdom$ value on the paths of the DFS.
The core idea of Lengauer-Tarjan algorithm is to fast compute the immediate dominators by the semidominators based on the following lemma. The details of the algorithm can be found in~\cite{LengauerT79}.

\begin{lemma}~\cite{LengauerT79}
Given $G=(V,E)$ and a source $s$, let $u\neq s$ and $v$ be the vertex with the minimum $sdom(v)$ among the vertices in the paths from $sdom(u)$ to $u$ (including $u$ but excluding $sdom(u)$), then we have $$idom(u)=\left\{\begin{matrix}
sdom(u) & \text{if } sdom(u)=sdom(v), \\
idom(v) & \text{otherwise}. 
\end{matrix}\right.$$
\label{lemma:sdom2idom}
\vspace{-3mm}
\end{lemma}


The time complexity of Lengauer-Tarjan algorithm is $O(m\cdot \alpha(m,n))$ which is almost linear, where $\alpha$ is the inverse function of Ackerman's function~\cite{ReelleZahlen}. 

\vspace{2mm}
\subsubsection{Compute the Expected Spread Decrease of Each Vertex}
\label{sec:new-es}

Following the above subsections, if a vertex $u$ is blocked, we can use the number of vertices in the subtree rooted at $u$ in the dominator tree to estimate the decrease of expected spread. 
Thus, using a depth-first search of each dominator tree, we can accumulate the decrease of expected spread for every vertex if it is blocked.

\begin{algorithm}[t]
    \SetVline 
    \SetFuncSty{textsf}
    \SetArgSty{textsf}
	\caption{DecreaseESComputation($G,s,\theta$)}
	\label{algo:blockes}
	\Input{a graph $G$, the source $s$ and the number of sampled graphs $\theta$}
	\Output{$\Delta [u]$ for each $u\in V(G)\setminus \{s\}$, which is the decrease of expected spread when $u$ is blocked}
	\State{$\Delta[\cdot]\leftarrow 0$}
	\For{$i\leftarrow 1$ to $\theta$}
	{
	    \State{Generate a sampled graph $g$ derived from $G$}
    	\State{$DT\leftarrow$ the dominator tree of $g$ which is constructed by Lengauer-Tarjan algorithm~\cite{LengauerT79}}
    	\State{$c[\cdot]\leftarrow$ the size of subtree in tree $DT$ when each vertex is the root}
    	\For{each $u\in V(g)\setminus \{s\}$}
    	{
    	    $\Delta[u]\leftarrow \Delta[u] + c[u]/\theta$
    	}
	}

	\Return{$\Delta[\cdot]$}
\end{algorithm}

\begin{theorem}
\revise{Let $s$ be a fixed vertex in graph $g$. For any vertex $u\in (V(g)\setminus \{s\})$, we have $\sigma^{\rightarrow u}(s,g)$ equals to the size of the subtree rooted at $u$ in the dominator tree of graph $g$.}
\label{theo:dominator-subtree}
\end{theorem}

\begin{proof}
\revise{Let $c[u]$ denote the size of the subtree with root $u$ in the dominator tree of graph $g$. Assume $v$ is in the subtree with $u$ (i.e., $u$ is the ancestor of $v$), from the definition of dominator, we have $v$ cannot be reached by $s$ when $u$ is blocked. Thus, $c[u]\le \sigma^{\rightarrow u}(s,g)$. If $v$ is not in the subtree, i.e., $u$ does not dominate $v$ in the graph, there is a path from $s$ to $v$ not through $u$, which means blocking $u$ will not affect the reachability from $s$ to $v$. We have $\sigma^{\rightarrow u}(s,g) =  c[u]$.}
\end{proof}

Thus, for each blocker $u$, we can estimate the decrease of the expected spread by the average size of the subtrees rooted at $u$ in the dominator trees of the sampled graphs.

\begin{example}
Considering the graph $G$ in Figure~\ref{fig:sample}, there are only three edges with propagation probabilities less than $1$ (i.e., $(v_5,v_8)$, $(v_9,v_8)$ and $(v_8,v_7)$), and the other edges will exist in any sampled graph. 
Figures~\ref{fig:sample-example-a}-\ref{fig:sample-example-d} depict all the possible sampled graphs.
For conciseness, we use the dotted edge $(v_8,v_7)$ to represent whether it may exist in a sampled graph or not (corresponding to two different sampled graphs, respectively).
When $(v_8,v_7)$ is not in the sampled graphs, as $p_{v_5,v_8}=0.5$ and $p_{v_9,v_8}=0.2$, Figures~\ref{fig:sample-example-a}, \ref{fig:sample-example-b}, \ref{fig:sample-example-c} and \ref{fig:sample-example-d} have $0.1, 0.4, 0.1$ and $0.4$ to exist, respectively. 
As $p_{v_8,v_7}=0.1$, $v_1$ can reach $8+0.1=8.1$ vertices in expectation in Figure~\ref{fig:sample-example-a}. Similarly, $v_1$ can reach $8.1,8.1$ and $7$ vertices (including $v_1$) in expectation in Figures~\ref{fig:sample-example-b}, \ref{fig:sample-example-c} and  \ref{fig:sample-example-d}, respectively.
Thus, the expected spread of graph $G$ is $8.1\times (0.1+0.4+0.1)+7\times 0.4=7.66$, which is the same as the result we compute in Example~\ref{example:compute}.

Figures~\ref{fig:sample-dt-example-a}-\ref{fig:sample-dt-example-d} show the corresponding dominator trees of the sampled graphs in Figure~\ref{fig:sample-example}. For vertex $v_5$, the expected sizes of the subtrees rooted at $v_5$ are $5.1,5.1,5.1$ and $4$ in the dominator trees, respectively. Thus, the blocking of $v_5$ will lead to $5.1\times(0.1+0.4+0.1)+4\times 0.4=4.66$ decrease of expected spread. As the sizes of subtrees of $v_2,v_3,v_4$ and $v_6$ are only $1$ in each dominator tree, blocking any of them will lead to $1$ expected spread decrease. 
Similarly, blocking $v_7$, $v_8$ and $v_9$ will lead to $0.66$, $0.06$ and $1.11$ expected spread decrease, respectively.
\end{example}

Algorithm~\ref{algo:blockes} shows the details for computing the decrease of expected spread of each vertex. We set $\Delta[\cdot]$ as $0$ initially (Line 1). Then we generate $\theta$ different sampled graphs derived from $G$ (Lines 2-3). For each sampled graph, we first construct the dominator tree through Lengauer-Tarjan (Line 4). Then we use a simple DFS to compute the size of each subtree. After computing the average size of the subtrees and recording it in $\Delta[\cdot]$ (Lines 6-7), we return $\Delta[\cdot]$ (Line 8).

As computing the sizes of the subtrees through DFS costs $O(m)$, Algorithm~\ref{algo:blockes} runs in $O(\theta\cdot m\cdot \alpha(m,n))$.

\subsection{Our AdvancedGreedy Algorithm}
\label{sec:advanced-greedy}

\begin{algorithm}[t]
    \SetVline 
    \SetFuncSty{textsf}
    \SetArgSty{textsf}
	\caption{AdvancedGreedy($G,s,b,\theta$)}
	\label{algo:fast-greedy}
	\Input{a graph $G$, the source $s$, budget $b$ and the number of sampled graphs $\theta$}
	\Output{the blocker set $B$}
    \State{$B\leftarrow$ empty}
	\For{$i\leftarrow 1$ to $b$}
	{
	    \State{$\Delta[\cdot] \leftarrow$ DecreaseESComputation($G[V\setminus B],s,\theta$)}
	    \State{$x\leftarrow  -1$}
	    \For{each $u\in V(G)\setminus (\{S\}\cup B)$}
	    {
	        \If{$x=-1$ or $\Delta[u] > \Delta[x]$}
	        {
	            \State{$x\leftarrow u$}
	        }
	    }
	    \State{$B\leftarrow B\cup \{x\}$}
	}
	\Return{$B$}
\end{algorithm}

Based on Section~\ref{sec:baseline} and Section~\ref{sec:candidates}, we propose AdvancedGreedy algorithm with high efficiency.
In the greedy algorithm, we aim to greedily find the vertex $u$ that leads to the largest decrease of expected spread. Algorithm~\ref{algo:blockes} can efficiently compute the expected spread decrease of every candidate blocker.
Thus, we can directly choose the vertex which can cause the maximum decrease of expected spread (i.e., the maximum average size of the subtrees in the dominator trees derived from sampled graphs) as the blocker.

Algorithm~\ref{algo:fast-greedy} presents the pseudo-code of our AdvancedGreedy algorithm. We start with the empty blocker set (Line 1). In each of the $b$ iterations (Line 2), we first estimate the decrease of expected spread of each vertex (Line 3), find the vertex $x$ such that $\Delta[x]$ is the largest as the blocker (Lines 2-7) and insert it to blocker set (Line 8). Finally, the algorithm returns the blocker set $B$ (Line 9).

\vspace{1mm}
\noindent \textbf{Comparison with Baseline.}
one round of MCS on $G$ will generate a graph $G'$ where $V(G')=V(G)$ and each edge in $E(G)$ will appear in $G'$ if the simulation picks this edge.
Thus, if we have $r=\theta$, our computation based on sampled graphs will not sacrifice the effectiveness, compared with MCS.
For efficiency, Algorithm~\ref{algo:fast-greedy} runs in $O(b\cdot \theta \cdot m\cdot \alpha(m,n) )$ and the time complexity of the saseline is $O(b\cdot r\cdot m\cdot n)$ (Algorithm~\ref{algo:greedy}). As $\alpha(m,n)$ is much smaller than $n$, our AdvancedGreedy algorithm has a lower time complexity without sacrificing the effectiveness, compared with the baseline algorithm.

\subsection{The GreedyReplace Algorithm}
\label{sec:greedyreplace}

\revise{
Some out-neighbors of the seed may be an essential part of the result while they may be missed by current greedy heuristics. Thus, we propose a new heuristic (GreedyReplace) which is to first select $b$ out-neighbors of the seed as the initial blockers, and then greedily replace a blocker with another vertex if the expected spread will decrease.}

\vspace{1mm}
\begin{example}
Considering the graph in Figure~\ref{fig:sample} with the seed $v_1$, Table~\ref{tab:motiv} shows the result of the Greedy algorithm and the result of only considering the out-neighbors as the candidate blockers (denoted as OutNeighbors). When $b=1$, Greedy chooses $v_5$ as the blocker because it leads to the largest expected spread decrease ($v_3,v_6,v_7,v_8$ and $v_9$ will not be influenced by $v_1$).
When $b = 2$, it further blocks $v_2$ or $v_4$ in the second round.
OutNeighbors only considers blocking $v_2$ and $v_4$. It blocks either of them when $b=1$, and blocks both of them when $b=2$.
\end{example}
\vspace{1mm}

\begin{table}[t]
\small
	\centering%
	\caption{Blockers and Their Expected Influence Spread}\label{tab:motiv}
	\resizebox{\linewidth}{!}{
	\begin{tabular}{c|c|c|c|c}
	\toprule
	\multirow{2}{*}{\centering Algorithm }& \multicolumn{2}{|c|}{$b=1$} & \multicolumn{2}{|c}{$b=2$}\\ \cline{2-5}
	 & $B$ & $\mathbb{E}(\cdot)$ & $B$ & $\mathbb{E}(\cdot)$ \\
	\midrule
	
	Greedy & $\{v_5\}$ & $3$ & $\{v_5,v_2$ or $v_4\}$ & $2$ \\ \hline
	OutNeighbors & $\{v_2$ or $v_4\}$ & $6.66$ &$\{v_2,v_4\}$ & $1$ \\ \hline
	GreedyReplace & $\{v_5\}$ & $3$ &$\{v_2,v_4\}$ & $1$ \\

	\bottomrule
	\end{tabular}
	}
\end{table}



In this example, we find that the performance of the Greedy algorithm is better than the OutNeighbors when $b$ is small, but its expected spread may become larger than OutNeighbors with the increase of $b$. 
As the budget $b$ can be either small or large in different applications, it is essential to further improve the heuristic algorithm.


\begin{algorithm}[t]
    \SetVline 
    \SetFuncSty{textsf}
    \SetArgSty{textsf}
	\caption{GreedyReplace($G,s,b,\theta$)}
	\label{algo:greedyreplace}
	\Input{a graph $G$, the source $s$, budget $b$ and the number of sampled graphs $\theta$}
	\Output{the blocker set $B$}
	\State{$CB\leftarrow N^{out}_s$}
    \State{$B\leftarrow$ empty}
	\For{$i\leftarrow 1$ to $\min\{d^{out}_s,b\}$}
	{
	    \State{$\Delta[\cdot] \leftarrow$ DecreaseESComputation($G[V\setminus B],s,\theta$)}
	    \State{$x\leftarrow  -1$}
	    \For{each $u\in CB$}
	    {
	        \If{$x=-1$ or $\Delta[u] > \Delta[x]$}
	        {
	            \State{$x\leftarrow u$}
	        }
	    }
	    \State{$CB\leftarrow CB\setminus \{x\}$}
	    \State{$B\leftarrow B\cup \{x\}$}
	}
	
	\For{each $u\in B$ with the reversing order of insertion}
	{
	    \State{$B\leftarrow B\setminus \{u\}$}
	    \State{$\Delta[\cdot] \leftarrow$ DecreaseESComputation($G[V\setminus B],s,\theta$)}
	    \State{$x\leftarrow  -1$}
	    \For{each $u\in V(G)\setminus B$}
	    {
	        \If{$x=-1$ or $\Delta[u] > \Delta[x]$}
	        {
	            \State{$x\leftarrow u$}
	        }
	    }
	    \State{$B\leftarrow B\cup \{x\}$}
	    \If{$u=x$}
	    {
	        \State{Break}
	    }
	}
	\Return{$B$}
\end{algorithm}


Due to the above motivation, based on Greedy and OutNeighbors, we propose the GreedyReplace algorithm to address their defects and combine the advantages.
We first greedily choose $b$ out-neighbors of the seed as the initial blockers. 
Then, we replace the blockers according to the reverse order of the out-neighbors' blocking order. 
As we can use Algorithm~\ref{algo:blockes} to compute the decrease of the expected spread of blocking any other vertex, in each round of replacement, we set all the vertices in $V(G)\setminus (\{s\}\cup B)$ as the candidates for replacement.
We will early terminate the replace procedure when the vertex to replace is current best blocker.

The expected spread of GreedyReplace is certainly not larger than the algorithm which only blocks the out-neighbors. Through the trade-off between choosing the out-neighbors and the replacement, the cooperation of the blockers is considered in GreedyReplace. 

\begin{example}
\label{example:move}
Considering the graph in Figure~\ref{fig:sample} with the seed $v_1$, Table~\ref{tab:motiv} shows the results of three algorithms.
When $b=1$, GreedyReplace first consider the out-neighbors as the candidate blockers and set $v_2$ or $v_4$ as the blocker. As blocking $v_5$ can achieve smaller influence spread than both $v_2$ and $v_4$, it will replace the blocker with $v_5$. When $b=2$, GreedyReplace first block $v_2$ and $v_4$, and there is no better vertex to replace. The expected spread is $1$. GreedyReplace achieves the best performance for either $b=1$ or $b=2$.
\end{example}

Algorithm~\ref{algo:greedyreplace} shows the pseudo-code of GreedyReplace. We first push all out-neighbors of the seed into candidate blocker set $CB$ (Line 1) and set blocker set empty initially (Line 2). For each round (Line 3), we choose the candidate blocker which leads to the largest expected spread decrease as the blocker (Lines 4-8) and then updates $CB$ and $B$ (Lines 9-10). 
Then we consider replacing the blockers in $B$ by the reversing order of their insertions (Line 11). We remove the replaced vertex from the blocker set (Line 12) and use Algorithm~\ref{algo:blockes} to compute the decrease of expected spread $\Delta[\cdot]$ for each candidate blocker (Line 13).
We use $x$ to record the vertex with the largest spread decrease computed so far (Line 14), by enumerating each of the candidate blockers (Lines 15-17). 
If the vertex to replace is current best blocker, we will early terminate the replacement (Lines 18-20). Algorithm~\ref{algo:greedyreplace} returns the set $B$ of $b$ blockers (Line 21).

The time complexity of GreedyReplace is $O(\min \{d^{out}_s,b\}\cdot \theta \cdot m\cdot \alpha(m,n))$. As the time complexity of Algorithm~\ref{algo:blockes} is mainly decided by the number of edges in the sampled graphs, thus in practice the time cost is much less than the worst case.



\subsection{\revise{Extension: IMIN Problem under Triggering Model}}

\revise{The triggering model is a generalization of both the IC model and the LT model~\cite{max-rr,first-max,Borgs-max}, which assumes that each vertex $u$ in $G$ is associated with a distribution $T(u)$ over subsets of $u$’s in-neighbors. 
For the given graph $G$, we can generate a sampled graph as follows: for each vertex $u$, we sample a triggering set of $u$ from $T(u)$, and remove each incoming edge of $u$ if the edge starts from a vertex not in the triggering set. With the sampled graphs, we can execute our AdvancedGreedy algorithm and GreedyReplace algorithm on them to solve the IMIN problem under the triggering model.}

\section{Experiments}
\label{sec:exp}

In this section, extensive experiments are conducted to validate the effectiveness and the efficiency of our algorithms.

\begin{table}[t]
	\caption{Statistics of Datasets}
	\centering%
	\small
	\resizebox{\linewidth}{!}{
	\begin{tabular}{l|ccccc}
	\toprule
		\bf{Dataset} &  \bf{$n$} & \bf{$m$} & \bf{$d_{avg}$} & \bf{$d_{max}$} & \bf{Type}  \\
		\midrule
		\textbf{E}mail\textbf{C}ore & 1,005 & 25,571 &  49.6 &544 & Directed\\
		\textbf{F}acebook & 4,039 & 88,234 & 43.7 & 1,045  & Undirected \\
		
		\textbf{W}iki-Vote & 7,115 & 103,689 & 29.1 & 1,167 & Directed\\

		\textbf{E}mail\textbf{A}ll & 265,214 & 420,045 & 3.2 & 7,636 & Directed\\
		\textbf{D}BLP & 317,080 & 1,049,866 & 6.6 & 343 & Undirected\\
		
		\textbf{T}witter & 81,306 & 1,768,149 & 59.5& 10,336  & Directed\\
		
		\textbf{S}tanford & 281,903 & 2,312,497 & 16.4 & 38,626 & Directed\\
		\textbf{Y}outube & 1,134,890 & 2,987,624 & 5.3 & 28,754 & Undirected\\
		
		\bottomrule
	\end{tabular}
	}
	\label{tab:dat}
\end{table}

\begin{figure}[t]
	\begin{minipage}{1\linewidth}
    \centering
    \includegraphics[width=\linewidth]{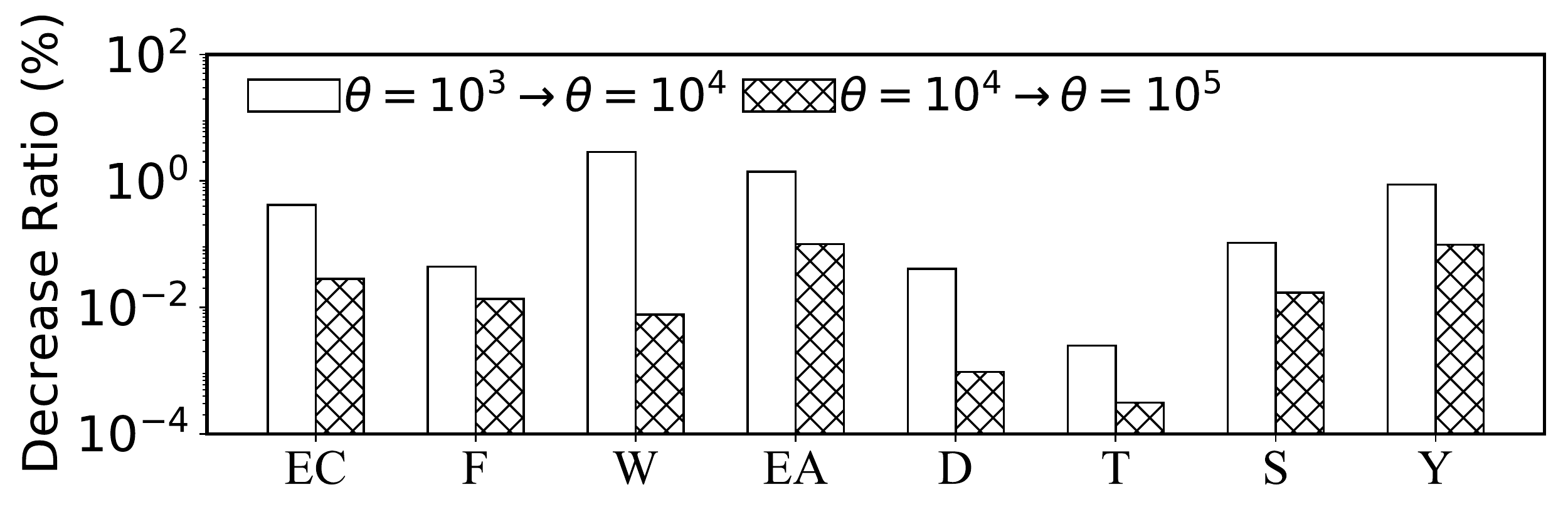}
    \caption{Expected Spread v.s. Number of Sampled Graphs}
    \label{fig:es-theta}
    \end{minipage}
 \vspace{4mm}
        
    \begin{minipage}{1\linewidth}
    \centering
    \includegraphics[width=\linewidth]{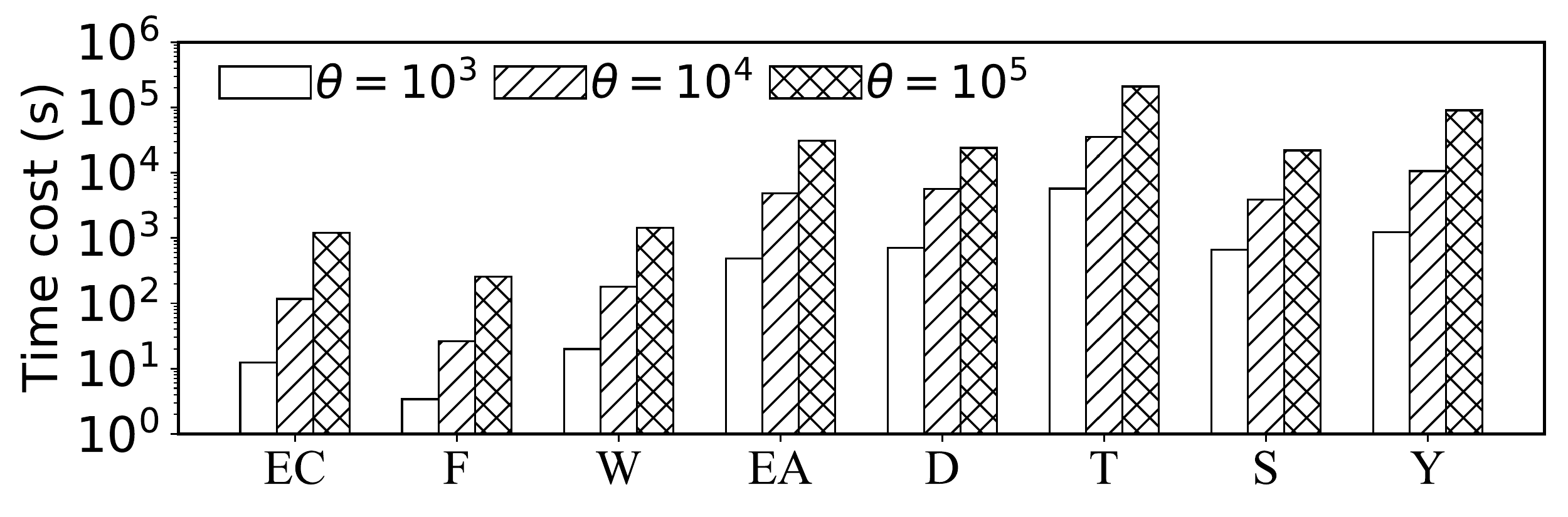}
    \caption{Running Time v.s. Number of Sampled Graphs}
    \label{fig:time-theta}
    \end{minipage}
    \vspace{-1mm}
\end{figure}



\subsection{Experimental Setting}

\label{sec:expset}

\noindent \textbf{Datasets.}
The experiments are conducted on $8$ datasets, obtained from SNAP (\url{http://snap.stanford.edu}). 
Table~\ref{tab:dat} shows the statistics of the datasets, ordered by the number of edges in each dataset, where $d_{avg}$ is the average vertex degree (the sum of in-degree and out-degree for each directed graph) and $d_{max}$ is the largest vertex degree.
For an undirected graph, we consider each edge as bi-directional.

\vspace{1mm}
\noindent \textbf{Propagation Models.} 
Following existing studies, e.g., \cite{maximization1,first-max}, we use two propagation probability models to assign the probability $p_{u,v}$ on each directed edge $(u,v)$: (i) Trivalency (TR) model, which assigns $p_{u,v}=TRI$ for each edge, where $TRI$ is uniformly selecting a value from $\{0.1,0.01,0.001\}$ \cite{maximization1,min-greedy-tree,JungHC12}; 
and (ii) Weighted cascade (WC) model, which assigns $p_{u,v}=1/d^{in}_v$~\cite{first-max,max-rr}.

\vspace{1mm}
\noindent \textbf{Setting.}
Unless otherwise specified, for Monte-Carlo Simulations, we set the number of rounds $r=10000$, and for our sampled graph based algorithm, we sample $10000$ graphs in the experiments. In each of our experiments, we independently execute each algorithm 5 times and report the average result. {\em By default, we terminate an algorithm if the running time reaches 24 hours}.  

\vspace{1mm}
\noindent \textbf{Environments.} The experiments are performed on a CentOS Linux serve (Release 7.5.1804) with Quad-Core Intel Xeon CPU (E5-2640 v4 @ 2.20GHz) and 128G memory. All the algorithms are implemented in C++. The source code is compiled by GCC(7.3.0) under O3 optimization.

\begin{table}[t]
    \centering
\caption{Exact v.s. GreedyReplace (TR Model)}\label{tab:exact-tr}
\resizebox{0.9\linewidth}{!}{
\begin{tabular}{c|c|c|c|c|c}
\toprule
\multirow{2}{*}{ b } & \multicolumn{3}{c|}{Expected Spread} & \multicolumn{2}{c}{Running Time (s)} \\ \cline{2-6}
& Exact & GR & Ratio & Exact & GR \\ \midrule
1 & 12.614 & 12.614 & 100\% & 3.07 & 0.12 \\
2 & 12.328 & 12.334 & 99.95\% & 130.91 & 0.21 \\
3 & 12.112 & 12.119 & 99.94\% & 3828.2 & 0.25 \\
4 & 11.889 & 11.903 & 99.88\% & 80050 & 0.33 \\ \bottomrule
\end{tabular}}
\end{table}

\begin{table}[t]
    \centering
\caption{Exact v.s. GreedyReplace (WC Model)}\label{tab:exact-wc}
\resizebox{0.9\linewidth}{!}{
\begin{tabular}{c|c|c|c|c|c}
\toprule
\multirow{2}{*}{ b } & \multicolumn{3}{c|}{Expected Spread} & \multicolumn{2}{c}{Running Time (s)} \\ \cline{2-6}
& Exact & GR & Ratio & Exact & GR \\ \midrule
1 & 11.185 & 11.185 & 100\% & 2.63 & 0.10 \\
2 & 11.077 & 11.078 & 99.99\% & 110.92 & 0.18 \\
3 & 10.997 & 10.998 & 99.99\% & 3284.0 & 0.23 \\
4 & 10.922 & 10.925 & 99.97\% & 69415 & 0.33 \\ \bottomrule
\end{tabular}}
\vspace{-1mm}
\end{table}

\vspace{1mm}
\noindent \textbf{Algorithms.}
In the experiments, we mainly compare our GreedyReplace algorithm and AdvancedGreedy algorithm with four basic algorithms (Exact, Rand, Out-degree and BaselineGreedy algorithm).

\noindent \underline{\texttt{Exact}}: identifies the optimal solution by searching all possible combinations of $b$ blockers, and uses Monte-Carlo Simulations with $r=10000$ to compute the expected spread of each candidate set of the blockers.

\noindent \underline{\texttt{Rand} (RA)}: randomly chooses $b$ blockers in the graph excluding the seeds. 

\noindent \underline{\texttt{OutDegree} (OD)}: selects $b$ vertices with the highest out-degrees as the blockers.



\noindent \underline{\texttt{BaselineGreedy} (BG)}: the state-of-the-art algorithm for the IMIN problem (Algorithm~\ref{algo:greedy})~\cite{min-greedy1,min-greedy2} that uses Monte-Carlo Simulations to compute the expected spread.

\noindent \underline{\texttt{AdvancedGreedy} (AG)}: Algorithm~\ref{algo:fast-greedy} which uses Algorithm~\ref{algo:blockes} to accelerate the BaselineGreedy algorithm.

\noindent \underline{\texttt{GreedyReplace} (GR)}: our GreedyReplace algorithm (Algorithm~\ref{algo:greedyreplace}).



\begin{table*}
\centering
\caption{Comparision with Other Heuristics (Expected Spread)}\label{table:heuristics}
\resizebox{\linewidth}{!}{
\begin{tabular}{c|cccc|cccc|cccc|cccc}
\toprule
\multirow{2}{*}{\centering b } & \multicolumn{4}{c|}{EmailCore (TR model)} & \multicolumn{4}{c|}{Facebook (TR model)} & \multicolumn{4}{c|}{Wiki-Vote (TR model)} & \multicolumn{4}{c}{EmailAll (TR model)} \\
& RA & OD & AG & GR & RA & OD & AG & GR & RA & OD & AG & GR & RA & OD & AG & GR \\
\midrule
20 & 354.88 & 230.10 & 220.59 & \bf{219.69} & 16.059 & 16.026 & 11.717 & \bf{11.691} & 512.62 & 325.51 & 131.30 & \bf{130.77} & 548.99 & 286.05 & 14.642 & \bf{13.640} \\
40 & 341.33 & 98.712 & 84.022 & \bf{83.823} & 16.037 & 16.019 & 10.416 & \bf{10.413} & 512.18 & 222.00 & 46.747 & \bf{43.898} & 546.94 & 221.97 & 10.319 & \bf{10.002} \\
60 & 325.13 & 47.249 & 35.085 & \bf{33.634} & 16.033 & 16.010 & 10.151 & \bf{10.149} & 507.11 & 138.60 & 25.514 & \bf{23.282} & 546.39 & 148.52 & 10 & 10 \\
80 & 304.90 & 30.277 & 19.001 & \bf{18.848} & 15.997 & 15.987 & 10.028 & \bf{10.026} & 501.49 & 32.646 & 17.332 & \bf{17.322} & 545.41 & 100.84 & 10 & 10 \\
100 & 283.54 & 22.696 & 13.640 & \bf{13.533} & 15.994 & 15.980 & 10.001 & 10.001 & 496.05 & 25.831 & 14.726 & \bf{14.518} & 544.59 & 55.398 & 10 & 10 \\
\bottomrule

\toprule
\multirow{2}{*}{\centering b } & \multicolumn{4}{c|}{DBLP (TR model)} & \multicolumn{4}{c|}{Twitter (TR model)} & \multicolumn{4}{c|}{Stanford (TR model)} & \multicolumn{4}{c}{Youtube (TR model)} \\
& RA & OD & AG & GR & RA & OD & AG & GR & RA & OD & AG & GR & RA & OD & AG & GR \\
\midrule
20 & 13.747 & 13.730 & 10.502 & \bf{10.499} & 16801 & 16610 & 16101 & \bf{16100} & 16.087 & 16.075 & 12.069 & \bf{10.48}3 & 14.774 & 14.762 & 14.743 & \bf{10.950} \\
40 & 13.739 & 13.725 & 10.079 & 10.079 & 16796 & 16470 & 15749 & \bf{15748} & 16.080 & 16.071 & 10.488 & \bf{10.234} & 14.773 & 14.755 & 10.075 & \bf{10.002} \\
60 & 13.737 & 13.721 & 10.012 & \bf{10.010} & 16786 & 16329 & 15447 & \bf{14972} & 16.071 & 16.040 & 10.136 & \bf{10.075} & 14.773 & 14.750 & 10 & 10 \\
80 & 13.720 & 13.714 & 10 & 10 & 16780 & 16175 & 14610 & \bf{14474} & 16.064 & 16.017 & 10.026 & \bf{10.019} & 14.767 & 14.742 & 10 & 10 \\
100 & 13.716 & 13.706 & 10 & 10 & 16771 & 16057 & 13619 & \bf{13181} & 16.052 & 15.989 & 10.009 & \bf{10.002} & 14.762 & 14.729 & 10 & 10 \\
\bottomrule

\toprule
\multirow{2}{*}{\centering b } & \multicolumn{4}{c|}{EmailCore (WC model)} & \multicolumn{4}{c|}{Facebook (WC model)} & \multicolumn{4}{c|}{Wiki-Vote (WC model)} & \multicolumn{4}{c}{EmailAll (WC model)} \\
& RA & OD & AG & GR & RA & OD & AG & GR & RA & OD & AG & GR & RA & OD & AG & GR \\
\midrule
20 & 82.605 & 54.907 & 53.516 & \bf{53.296} & 21.482 & 21.362 & 14.588 & \bf{14.554} & 24.102 & 22.660 & 17.765 & \bf{17.701} & 13.330 & 11.493 & 10.720 & \bf{10.455} \\
40 & 75.990 & 44.710 & 40.199 & \bf{40.093} & 21.456 & 21.360 & 12.425 & \bf{12.418} & 23.971 & 21.696 & 15.258 & \bf{15.222} & 13.234 & 11.447 & 10.111 & \bf{10.086} \\
60 & 69.947 & 37.561 & 31.891 & \bf{31.784} & 21.429 & 21.297 & 11.194 & \bf{11.187} & 23.899 & 20.409 & 13.749 & \bf{13.743} & 13.217 & 11.408 & 10 & 10 \\
80 & 64.154 & 32.580 & 26.094 & \bf{26.073} & 21.417 & 21.176 & 10.476 & \bf{10.474} & 23.763 & 12.798 & 12.654 & \bf{12.574} & 13.188 & 11.385 & 10 & 10 \\
100 & 57.170 & 24.959 & 21.926 & \bf{21.899} & 21.395 & 21.056 & 10.013 & \bf{10.012} & 23.757 & 12.711 & 12.138 & \bf{12.129} & 13.070 & 11.324 & 10 & 10 \\
\bottomrule

\toprule
\multirow{2}{*}{\centering b } & \multicolumn{4}{c|}{DBLP (WC model)} & \multicolumn{4}{c|}{Twitter (WC model)} & \multicolumn{4}{c|}{Stanford (WC model)} & \multicolumn{4}{c}{Youtube (WC model)} \\ 
& RA & OD & AG & GR & RA & OD & AG & GR & RA & OD & AG & GR & RA & OD & AG & GR \\
\midrule
20 & 118.23 & 118.17 & 32.602 & \bf{32.601} & 259.45 & 235.73 & 199.40 & \bf{198.76} & 25.803 & 25.800 & 11.742 & \bf{11.740} & 25.663 & 25.368 & 10.152 & \bf{10.113} \\
40 & 118.16 & 117.81 & 18.429 & \bf{18.409} & 258.86 & 226.28 & 170.20 & \bf{168.64} & 25.790 & 25.779 & 10.435 & \bf{10.398} & 25.585 & 25.330 & 10.012 & \bf{10.006} \\
60 & 118.09 & 117.73 & 11.869 & \bf{11.867} & 257.75 & 214.90 & 144.61 & \bf{144.23} & 25.777 & 25.771 & 10.119 & \bf{10.101} & 25.446 & 25.149 & 10 & 10 \\
80 & 117.97 & 117.59 & 10 & 10 & 255.48 & 204.17 & 129.53 & \bf{128.75} & 25.685 & 25.633 & 10.005 & \bf{10.002} & 25.346 & 25.113 & 10 & 10 \\
100 & 117.94 & 117.43 & 10 & 10 & 254.04 & 196.40 & 114.11 & \bf{112.07} & 25.657 & 25.603 & 10.002 & \bf{10.000} & 25.277 & 25.052 & 10 & 10 \\
\bottomrule
\end{tabular} }
\end{table*}

\subsection{Effectiveness}

\noindent
{\bf Varying the Number of Sampled Graphs.} In Figure~\ref{fig:es-theta} and Figure~\ref{fig:time-theta}, we vary $\theta$ (i.e., the number of sampled graphs for choosing the blocker in each round) from $10^3$, $10^4$ to $10^5$, and report the expected spread and running time of our GR algorithm. 
We evaluate on all datasets under the TR model by setting the blocker budget to $20$ and randomly selecting $10$ seed vertices.
We show the decrease ratio of expected spread of three $\theta$ values in Figure~\ref{fig:es-theta}, because the absolute differences of expected spreads are quite small.
The largest decrease ratio from $\theta=10^3$ to $\theta=10^4$ is only $2.89\%$, and the largest decrease ratio from $\theta=10^4$ to $\theta=10^5$ is less than $0.1\%$.
Figure~\ref{fig:time-theta} shows the running time gradually increases when $\theta$ increases. 
According to above results, we set $\theta=10^4$ in all the experiments for a good trade-off between the time cost and the accuracy. 

\vspace{1mm}
\noindent \textbf{Comparison with the Exact Algorithm.}
We also compare the result of GR with the Exact algorithm which identifies the optimal $b$ blockers by enumerating all possible combinations of $b$ vertices. Due to the huge time cost of Exact, we extract small datasets by iteratively extracting a vertex and all its neighbors, until the number of extracted vertices reaches $100$. For \texttt{EmailCore} under both WC and TR models, we extract $5$ such subgraphs. We randomly choose $10$ vertices as the seeds. As the graph is small, we can use the exact computation of the expected spread~\cite{MaeharaSI17} for comparison between Exact and GR.
Tables~\ref{tab:exact-tr} and \ref{tab:exact-wc} show that the expected spread of GR under both two influence propagation models is very close to the results of Exact while GR is faster than Exact by up to $6$ orders of magnitude.

\begin{figure}[t]
	\begin{minipage}{1\linewidth}
    \centering
    \includegraphics[width=\linewidth]{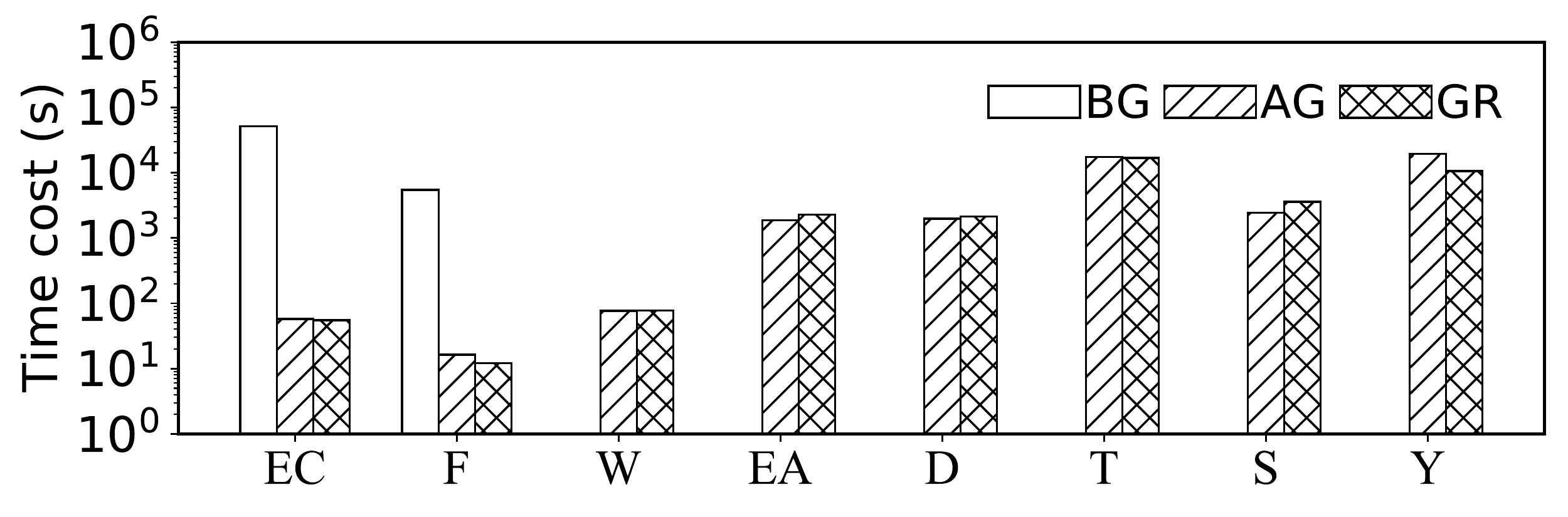}
    \caption{Time Cost of Different Algorithms under TR Model}
    \label{fig:time-tr}
    \end{minipage}
    \vspace{4mm}
    
	\begin{minipage}{1\linewidth}
    \centering
    \includegraphics[width=\linewidth]{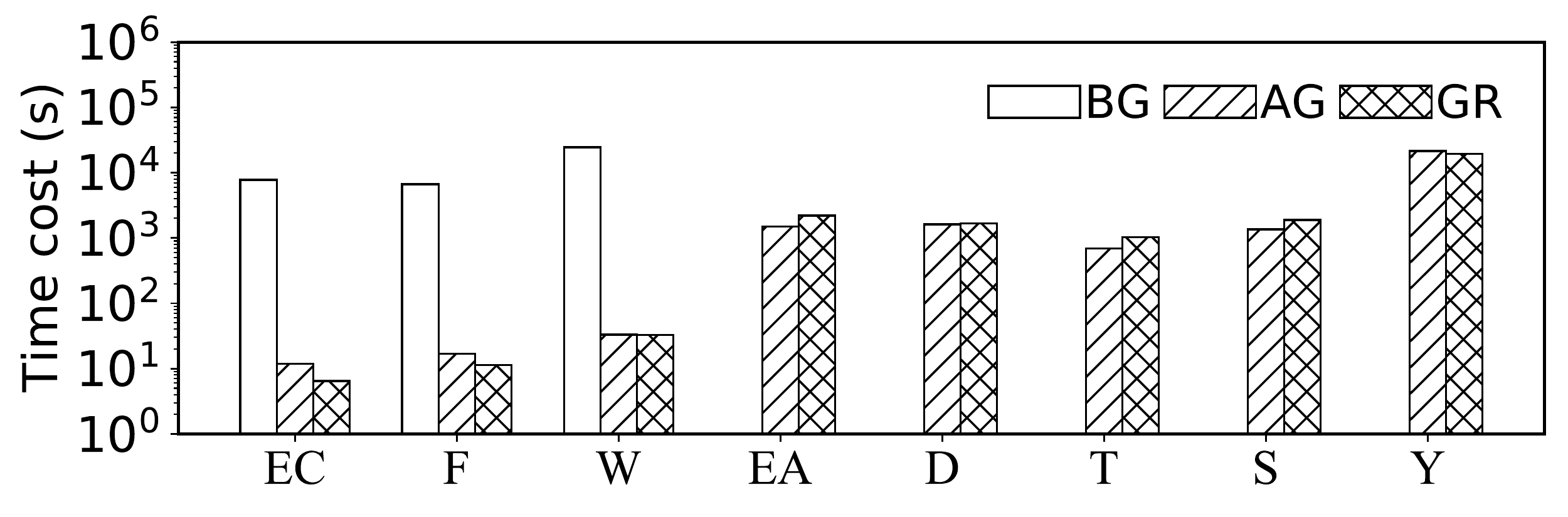}
    \caption{Time Cost of Different Algorithms under WC Model}
    \label{fig:time-wc}
    \end{minipage}
    \vspace{-1mm}
\end{figure}

\vspace{1mm}
\noindent \textbf{Comparison with Other Heuristics.}
As discussed in Section~\ref{sec:advanced-greedy}, the effectiveness (expected spread) of the AdvancedGreedy algorithm is the same as the BaselineGreedy algorithm. Thus, in this experiment, we compare Rand (RA), Out-degree (OD) with our AdvancedGreedy (AG) and GreedyReplace (GR) algorithms in Table~\ref{table:heuristics}. We first randomly select $10$ vertices as the seeds and vary the budget $b$ from $20,40,60,80$ to $100$. We repeat this process by $5$ times and report the average expected spread with the resulting blockers (the expected spread is computed by Monte-Carlo Simulations with $10^5$ rounds) on all datasets.
The results show that our GR algorithm always achieves the best result in both two propagation models (the smallest spread with different budgets), compared with RA, OD and AG. Besides, with the increase of budget $b$, the influence spread is better limited in GR.
The results verify that it is effective to first limit the candidate blockers in the out-neighbors of the seeds and then replace the candidates to improve the result.


\subsection{Efficiency}

\begin{figure}[t]
\centering
\begin{subfigure}[h]{0.48\linewidth}
    \includegraphics[width=1\columnwidth]{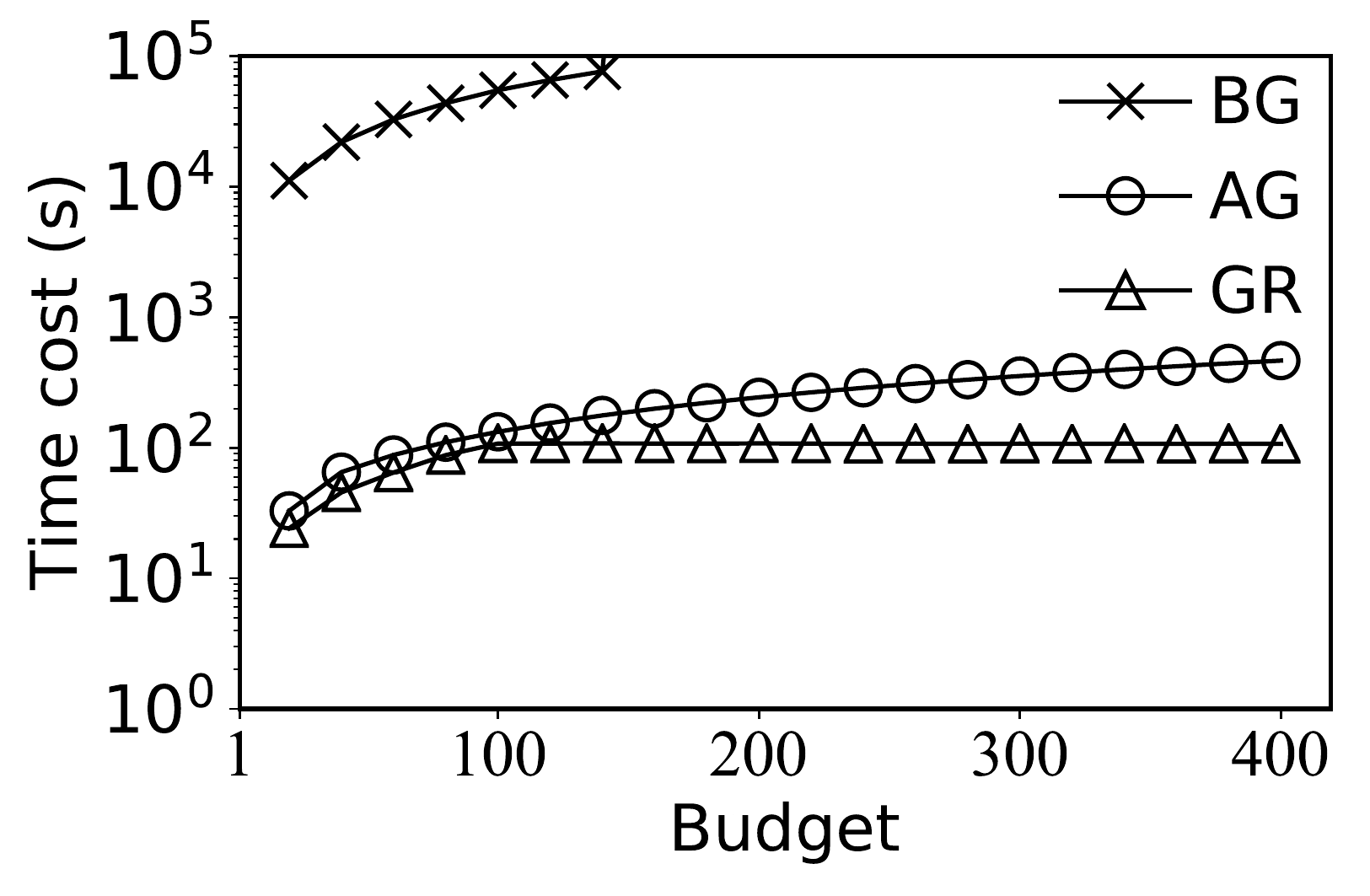}
    \caption{\texttt{Facebook} under TR model}\label{fig:time-budget-a}
\end{subfigure}
\begin{subfigure}[h]{0.48\linewidth}
    \includegraphics[width=1\columnwidth]{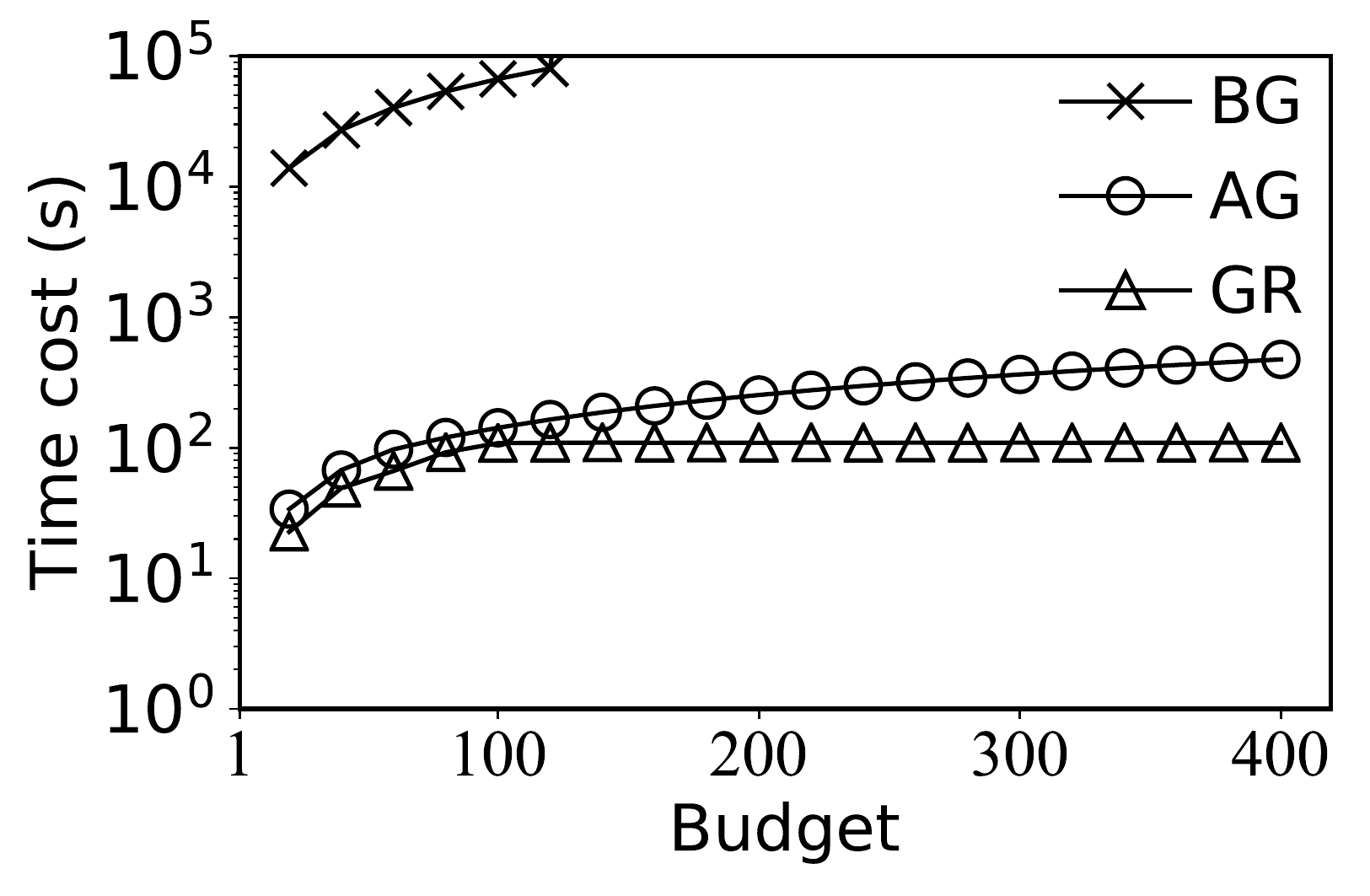}
    \caption{\texttt{Facebook} under WC model}\label{fig:time-budget-b}
\end{subfigure}
\begin{subfigure}[h]{0.48\linewidth}
    \includegraphics[width=1\columnwidth]{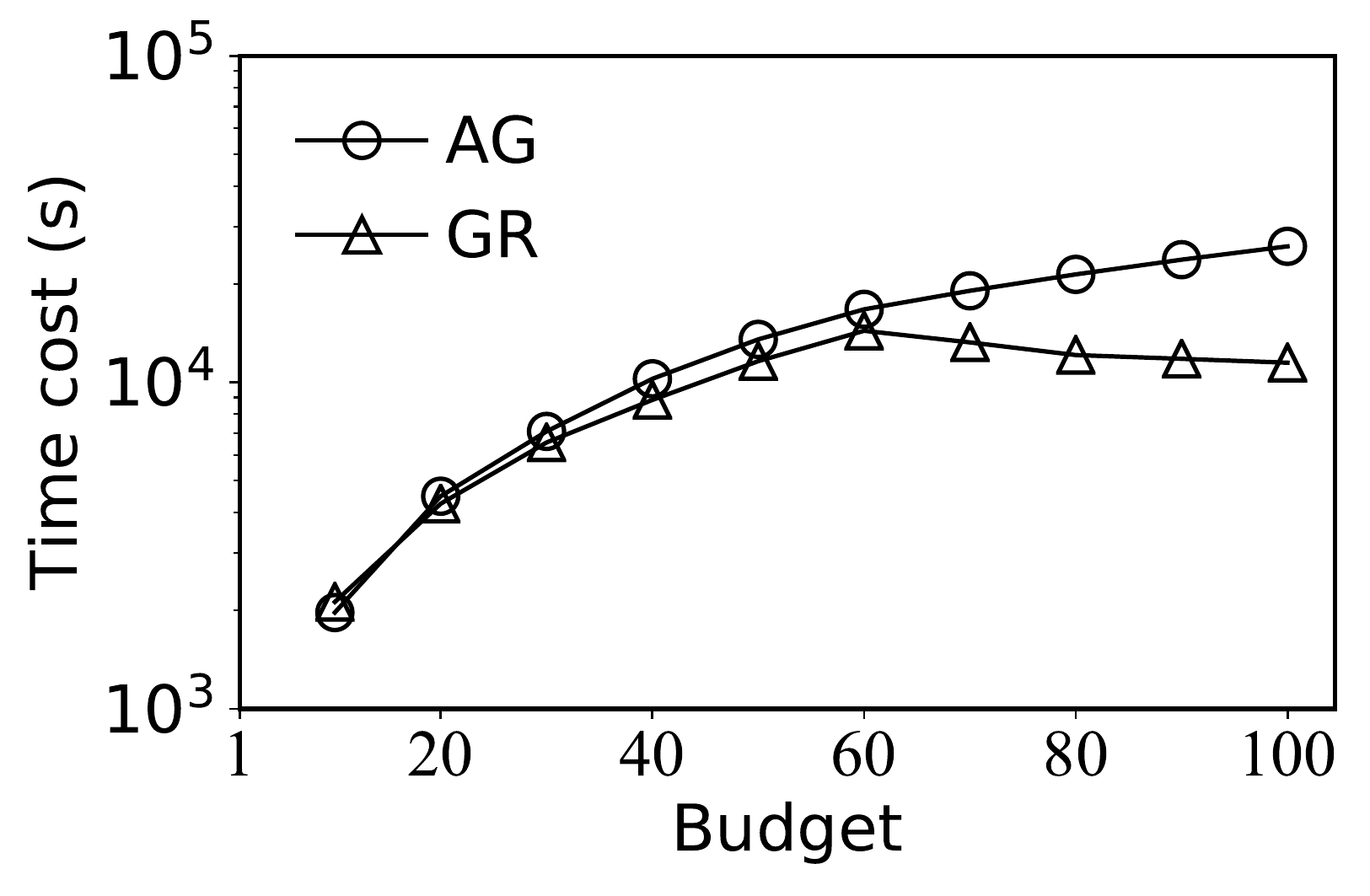}
    \caption{\texttt{DBLP} under TR model}\label{fig:time-budget-c}
\end{subfigure}
\begin{subfigure}[h]{0.48\linewidth}
    \includegraphics[width=1\columnwidth]{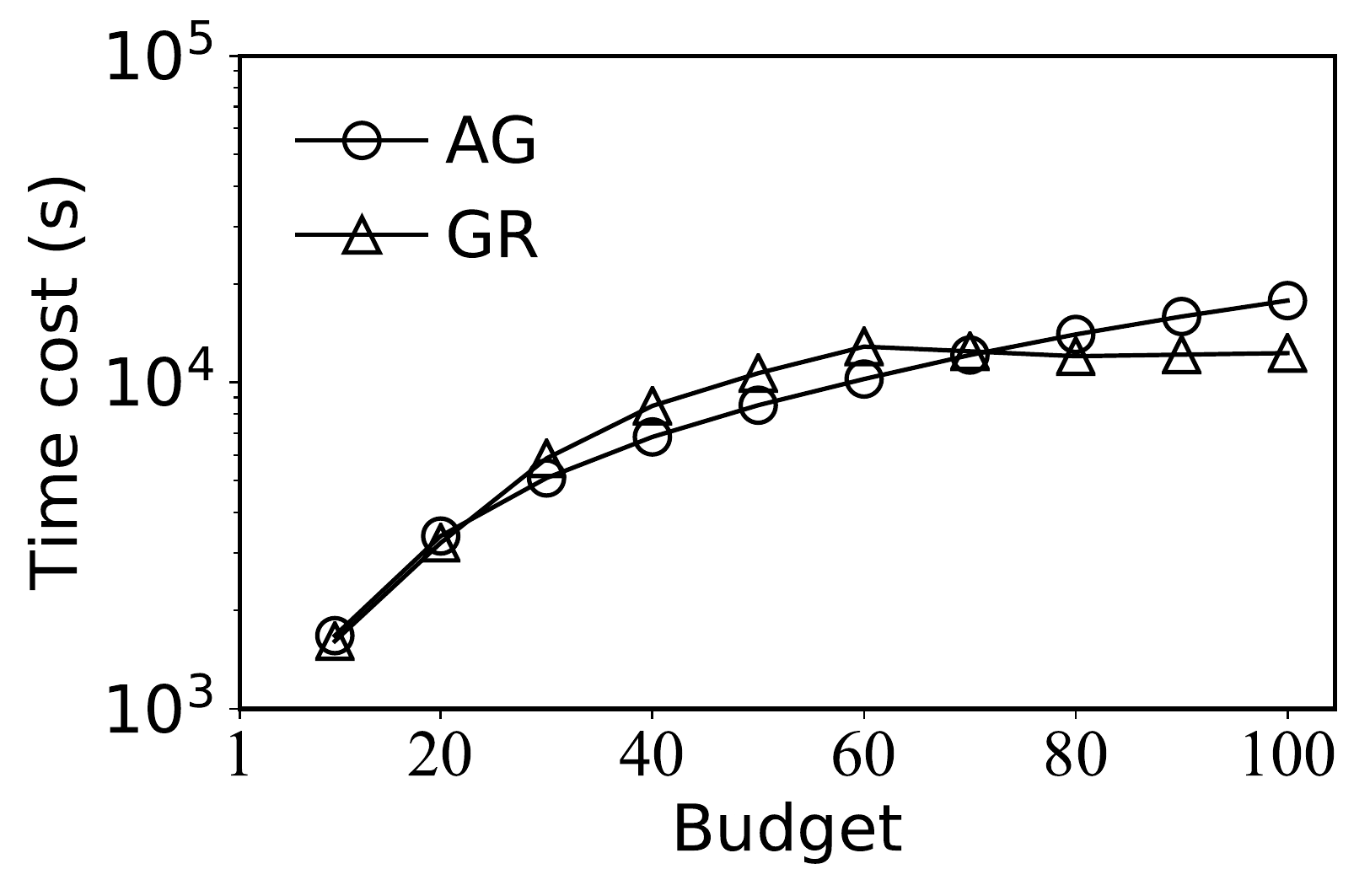}
    \caption{\texttt{DBLP} under WC model}\label{fig:time-budget-d}
\end{subfigure}
\caption{Running Time v.s. Budget}
\label{fig:time-budget}
    \vspace{-1mm}
\end{figure}

\noindent \textbf{Time Cost of Different Algorithms.}
Here we compare the running time of BG, AG and GR. We set the budget $b$ to $10$ due to the huge computation cost of the BG algorithm. Figures~\ref{fig:time-tr} and \ref{fig:time-wc} show the results in all dataset under the two propagation models. In $6$ datasets (resp. $5$ datasets) under the TR model (resp. WC model), BG cannot return results within the given time limit (i.e., $24$ hours).
The results show our AG and GR algorithms significantly outperform BG by at least $3$ orders of magnitude in runtime, and the gap can be larger on larger datasets which is consistent with the analysis of time complexities (Section~\ref{sec:advanced-greedy}). Besides, the time cost of GR is close to AG.

\vspace{1mm}
\noindent
{\bf Varying the Budget.}
Here we present the running time of \texttt{Facebook} and \texttt{DBLP} datasets by given different budgets in Figure~\ref{fig:time-budget}.
The running time of AG may decrease when the budget becomes larger due to the early termination applied in Algorithm~\ref{algo:greedyreplace} (Lines 19-20).
It is clear that AG and GR have much higher efficiency than BG, and the gap between them becomes even larger as the budget increases. We also find that the running time of AG is close to GR. AG may be faster than GR when the budget is small but GR performs better on the running time when the budget increases.

\begin{figure}[t]
	\begin{minipage}{1\linewidth}
    \centering
    \includegraphics[width=\linewidth]{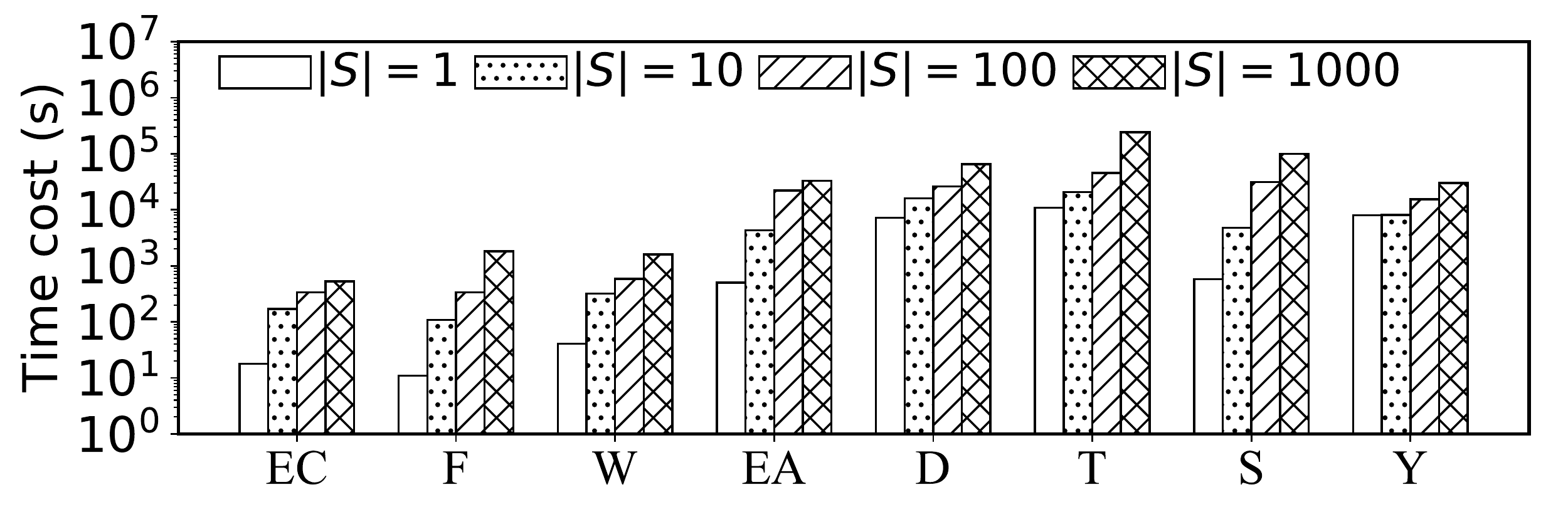}
    \caption{Running Time v.s. Number of Seeds (TR Model)}
    \label{fig:sca-tr}
    \end{minipage}
     \vspace{4mm}
       
 	\begin{minipage}{1\linewidth}
    \centering
    \includegraphics[width=\linewidth]{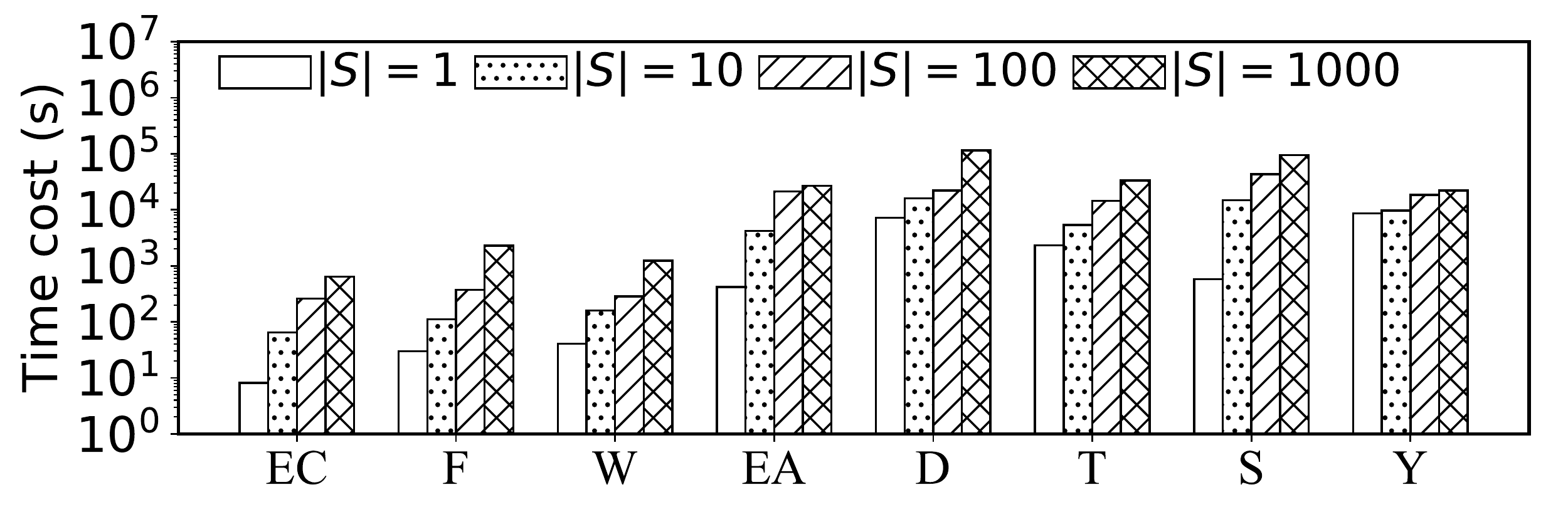}
    \caption{Running Time v.s. Number of Seeds (WC Model)}
    \label{fig:sca-wc}
    \end{minipage}   
    \vspace{-1mm}
\end{figure}

\vspace{1mm}
\noindent \textbf{Scalability.}
In Figure~\ref{fig:sca-tr} and Figure \ref{fig:sca-wc}, we test the scalability of our GR algorithm. We set the budget to $100$ and vary the number of seeds from $1,10,100$ to $1000$. We report the average time cost of the GR algorithm by executing it 5 times. We can see that the running time becomes larger as the number of seeds increases. It is because a large number of seeds leads to a wider influence spread (a larger size of sampled graphs), and the running time of Algorithm~\ref{algo:blockes} is highly related to the size of sampled graphs. 
We can also find that the increasing ratio of the running time is much less than the increasing ratio of the number of seeds, which validates that GR is scalable to handle the scenarios when the number of seeds is large.
\label{sec:exp-im}

\section{Conclusion}

Minimizing the influence of misinformation is critical for a social network to serve as a reliable platform. 
In this paper, we systematically study the influence minimization problem to blocker $b$ vertices such that the influence spread of a given seed set is minimized. 
We prove the problem is NP-hard and hard to approximate.
A novel spread estimation algorithm is first proposed to largely improve the efficiency of state-of-the-art without sacrificing the effectiveness. Then we propose the GreedyReplace algorithm to refine the effectiveness of the greedy method by considering a new heuristic. 
Extensive experiments on 8 real-life datasets verify that our GreedyReplace and AdvancedGreedy algorithms largely outperform the competitors. 
For future work, it is interesting to adapt our algorithm to other diffusion models and efficiently address other influence problems.


\section*{Acknowledgments}
Fan Zhang is partially supported by NSFC 62002073, NSFC U20B2046, Guangzhou Research Foundation (202102020675) and HKUST-GZU JRF (YH202202).
Wenjie Zhang is partially supported by ARC FT210100303 and ARC DP200101116. 
Xuemin Lin is partially supported by Guangdong Basic and Applied Basic Research Foundation (2019B1515120048). 

\bibliographystyle{IEEEtran}
\bibliography{ref}



\balance
\end{document}